\documentclass[11pt]{article}
\usepackage{fullpage}
\usepackage{refcount}
\usepackage{amsmath,amsfonts,amssymb,amsthm}
\usepackage{algorithm,algpseudocode,float}
\usepackage{xspace}
\usepackage[usenames]{color}

\newcommand{\supp}{\mathop{support}}
\newcommand{\suppfind}[1]{support-finding$_{{#1}}$}

\newtheorem{theorem}{Theorem}
\newtheorem{lemma}{Lemma}

\newtheorem{corollary}{Corollary}
\newtheorem{claim}{Claim}
\newtheorem{remark}{Remark}
\DeclareMathOperator*{\E}{\mathbb{E}}
\let\Pr\relax
\DeclareMathOperator*{\Pr}{\mathbb{P}}

\newcommand{\success}{\textsf{SUCC}\xspace}
\newcommand{\diane}{\mathsf{Diane}}
\newcommand{\enc}{\textsf{ENC}\xspace}
\newcommand{\dec}{\textsf{DEC}\xspace}
\newcommand{\aug}{\mathbf{AugIndex}\xspace}
\newcommand{\s}{\textsf{s}\xspace}
\newcommand{\R}{\mathbb{R}}
\newcommand{\F}{\mathbb{F}}

\newcommand{\sketch}{\mathsf{Alice}}
\newcommand{\query}{\mathsf{Bob}}
\newcommand{\eps}{\varepsilon}
\newcommand{\ur}{\mathbf{UR}\xspace}
\newcommand{\randcom}{\mathbf{R}}
\newcommand{\findup}[1]{\textsf{FindDuplicate}$({#1})$\xspace}
\newcommand{\poly}{{\mathrm{poly}}}

\title{Optimal lower bounds for universal relation, and for samplers and finding duplicates in streams\footnote{This paper is a merger of \cite{NelsonPW17},
    and of work of Kapralov, Woodruff, and Yahyazadeh.}}
\author{
  Michael Kapralov\thanks{EPFL. \texttt{michael.kapralov@epfl.ch}.}
  \and Jelani Nelson\thanks{Harvard University. \texttt{minilek@seas.harvard.edu}. Supported by NSF grant IIS-1447471 and
   CAREER award CCF-1350670, ONR Young Investigator award N00014-15-1-2388, and a Google Faculty Research Award.}
  \and Jakub Pachocki\thanks{OpenAI. \texttt{jakub@openai.com}. Work done while affiliated with Harvard University, under the support of ONR grant N00014-15-1-2388.}
  \and Zhengyu Wang\thanks{Harvard University. \texttt{zhengyuwang@g.harvard.edu}. Supported by NSF grant CCF-1350670.}
  \and David P. Woodruff\thanks{IBM Research Almaden. \texttt{dpwoodru@us.ibm.com.}}
  \and Mobin Yahyazadeh\thanks{Sharif University of Technology. \texttt{mn.yahyazadeh@gmail.com}. Work done while an intern at EPFL.}}

\begin{document}

\setcounter{page}{0}

\maketitle

\thispagestyle{empty}

\begin{abstract}
In the communication problem $\ur$ (universal relation) \cite{KarchmerRW95}, Alice and Bob respectively receive $x, y \in\{0,1\}^n$ with the promise that $x\neq y$. The last player to receive a message must output an index $i$ such that $x_i\neq y_i$. We prove that the randomized one-way communication complexity of this problem in the public coin model is exactly $\Theta(\min\{n,\log(1/\delta)\log^2(\frac n{\log(1/\delta)})\})$ for failure probability $\delta$. Our lower bound holds even if promised $\mathop{support}(y)\subset \mathop{support}(x)$. As a corollary, we obtain optimal lower bounds for $\ell_p$-sampling in strict turnstile streams for $0\le p < 2$, as well as for the problem of finding duplicates in a stream. Our lower bounds do not need to use large weights, and hold even if promised $x\in\{0,1\}^n$ at all points in the stream. 

We give two different proofs of our main result. The first proof demonstrates that any algorithm $\mathcal A$ solving sampling problems in turnstile streams in low memory can be used to encode subsets of $[n]$ of certain sizes into a number of bits below the information theoretic minimum. Our encoder makes adaptive queries to $\mathcal A$ throughout its execution, but done carefully so as to not violate correctness. This is accomplished by injecting random noise into the encoder's interactions with $\mathcal A$, which is loosely motivated by techniques in differential privacy. Our correctness analysis involves understanding the ability of $\mathcal A$ to correctly answer adaptive queries which have positive but bounded mutual information with $\mathcal A$'s internal randomness, and may be of independent interest in the newly emerging area of adaptive data analysis with a theoretical computer science lens. Our second proof is via a novel randomized reduction from Augmented Indexing \cite{MiltersenNSW98} which needs to interact with $\mathcal A$ adaptively. To handle the adaptivity we identify certain likely interaction patterns and union bound over them to guarantee correct interaction on all of them. To guarantee correctness, it is important that the interaction hides some of its randomness from $\mathcal A$ in the reduction.
\end{abstract}

\newpage

\section{Introduction}\label{sec:intro}
In turnstile $\ell_0$-sampling, a vector $z\in\R^n$ starts as the zero vector and receives coordinate-wise updates of the form ``$z_i \leftarrow z_i + \Delta$'' for $\Delta\in\{-M,-M+1,\ldots,M\}$. During a query, one must return a uniformly random element from $\supp(x) = \{i : z_i\neq 0\}$. The problem was first defined in \cite{FrahlingIS08}, where a data structure (or ``sketch'') for solving it was used to estimate the Euclidean minimum spanning tree, and to provide $\eps$-approximations of a point set $P$ in a geometric space (that is, one wants to maintain a subset $S\subset P$ such that for any set $R$ in a family of bounded VC-dimension, such as the set of all axis-parallel rectangles, $||R\cap S|/|S| - |R\cap P|/|P|| < \eps$). Sketches for $\ell_0$-sampling were also used to solve various dynamic graph streaming problems in \cite{AhnGM12a} and since then have been crucially used in almost all known dynamic graph streaming algorithms\footnote{\label{specfootnote}The spectral sparsification algorithm of \cite{KapralovLMMS14} is a notable exception.}, such as for: connectivity, $k$-connectivity, bipartiteness, and minimum spanning tree \cite{AhnGM12a}, subgraph counting, minimum cut, and cut-sparsifier and spanner computation \cite{AhnGM12b}, spectral sparsifiers \cite{AhnGM13}, maximal matching \cite{ChitnisCHM15}, maximum matching \cite{AhnGM12a,BuryS15,Konrad15,AssadiKLY16,ChitnisCEHMMV16,AssadiKL17}, vertex cover \cite{ChitnisCHM15,ChitnisCEHMMV16}, hitting set, $b$-matching, disjoint paths, $k$-colorable subgraph, and several other maximum subgraph problems \cite{ChitnisCEHMMV16}, densest subgraph \cite{BhattacharyaHNT15,McGregorTVV15,EsfandiariHW16}, vertex and hyperedge connectivity \cite{GuhaMT15}, and graph degeneracy \cite{FarachColtonT16}. For an introduction to the power of $\ell_0$-sketches in designing dynamic graph stream algorithms, see the recent survey of McGregor \cite[Section 3]{McGregor14}. Such sketches have also been used outside streaming, such as in distributed algorithms \cite{HegemanPPSS15,Pandurangan0S16} and data structures for dynamic connectivity \cite{KapronKM13,Wang15,GibbKKT15}.

Given the rising importance of $\ell_0$-sampling in algorithm design, a clear task is to understand the exact complexity of this problem. The work \cite{JowhariST11} gave an $\Omega(\log^2 n)$-bit space lower bound for data structures solving even the case $M=1$ which fail with constant probability, and otherwise whose query responses are $(1/3)$-close to uniform in statistical distance. They also gave an upper bound for $M \le \poly(n)$ with failure probability $\delta$, which in fact gave $\min\{\|z\|_0, \Theta(\log(1/\delta))\}$ uniform samples from the support of $z$, using space $O(\log^2 n \log(1/\delta))$ (here $\|z\|_0$ denotes $|\supp(z)|$). Thus we say their data structure actually solves the harder problem of $\ell_0$-sampling$_k$ for $k =\Theta(\log(1/\delta))$ with failure probability $\delta$, where in $\ell_0$-sampling$_k$ the goal is to recover $\min\{\|z\|_0, k\}$ uniformly random elements, without replacement, from $\supp(z)$.  The upper and lower bounds in \cite{JowhariST11} thus match up to a constant factor for $k = 1$ and $\delta$ a constant. We note though in many settings, even if the final application desires constant failure probability, $\ell_0$-sampling$_k$ with either failure probability $o(1)$ or $k>1$ (or both) is needed as a subroutine (see Figure~\ref{fig:table}).

\begin{figure}
\begin{center}
\begin{tabular}{|c|c|c|c|c|}
\hline
reference & problem & distribution & $k>1$? & $\delta = o(1)$?\\
\hline
\cite{FrahlingIS08} & Euclidean minimum spanning tree & $\ell_0$ & yes &\\
\hline
\cite{AhnGM12a} & connectivity\footnotemark & any & & yes\\
\hline
\cite{AhnGM12a} & $k$-connectivity\footnotemark[\getrefnumber{notejst}] & any & & yes\\
\hline
\cite{AhnGM12a} & bipartiteness\footnotemark[\getrefnumber{notejst}] & any & & yes\\
\hline
\cite{AhnGM12a} & minimum spanning tree & any & & yes\\
\hline
\cite{AhnGM12b} & subgraph counting & $\ell_0$ & yes & \\
\hline
\cite{AhnGM12b} & minimum cut & any &  & yes\\
\hline
\cite{AhnGM12b} & cut sparsifiers & any &  & yes\\
\hline
\cite{AhnGM12b} & spanners & any & yes & yes\\
\hline
\cite{AhnGM12b} & spectral sparsifiers & any &  & yes\\
\hline
\cite{ChitnisCHM15} & maximal matching & $\ell_0$ & yes & yes\\
\hline
\cite{BuryS15} & maximum matching (unweighted) & $\ell_0$ & yes & \\
& maximum matching (weighted) & $\ell_0$ & yes & yes \\
\hline
\cite{Konrad15} & maximum matching & any & yes & yes \\
\hline
\cite{AssadiKLY16} & maximum matching & $\ell_0$ & & yes \\
\hline
\cite{AssadiKL17} & maximum matching & $\ell_0$ & & yes \\
\hline
\cite{ChitnisCEHMMV16} & maximum matching & $\ell_0$ & & yes \\
& vertex cover &  & &  \\
& hitting set &  & &  \\
& $b$-matching &  & &  \\
& disjoint paths &  & &  \\
& $k$-colorable subgraph & & & \\
\hline
\cite{BhattacharyaHNT15} & densest subgraph & $\ell_0$ & & yes \\
\hline
\cite{McGregorTVV15} & densest subgraph & $\ell_0$ & yes & yes \\
\hline
\cite{EsfandiariHW16} & densest subgraph & $\ell_0$ & yes & \\
\hline
\cite{GuhaMT15} & vertex connectivity & any & & yes\\
 & hyperedge connectivity &  & & \\
\hline
\cite{FarachColtonT16} & graph degeneracy & $\ell_0$ & yes & \\
\hline
\end{tabular}
\caption{Guarantees needed by various works using samplers as subroutines. The last two columns indicate whether the work needs to use a sampler that returns $k$ samples at a time when queried for some $k>1$, or for some subconstant failure probability $\delta$ even to achieve failure probability $1/3$ in the main application. The ``distribution'' column indicates the output distribution needed from the sampler for the application (``any'' means a \suppfind{} subroutine is sufficient, i.e.\ it suffices for a query to return any index $i$ for which $z_i\neq 0$).}\label{fig:table}
\end{center}
\end{figure}

\paragraph{Universal relation.} The work of \cite{JowhariST11} obtains its lower bound for $\ell_0$-sampling (and some other problems) via reductions from {\em universal relation} ($\ur$). The problem $\ur$ was first defined in \cite{KarchmerRW95} and arose in connection with work of Karchmer and Wigderson on circuit depth lower bounds \cite{KarchmerW90}. For $f:\{0,1\}^n\rightarrow\{0,1\}$, $D(f)$ is the minimum depth of a fan-in $2$ circuit over the basis $\{\neg, \vee, \wedge\}$ computing $f$. Meanwhile, the (deterministic) communication complexity $C(f)$ is defined as the minimum number of bits that need to be communicated in a correct protocol for Alice and Bob to solve the following communication problem: Alice receives $x\in f^{-1}(0)$ and Bob receives $y\in f^{-1}(1)$ (and hence in particular $x\neq y$), and they must both agree on an index $i\in[n]$ such that $x_i\neq y_i$. It is shown in \cite{KarchmerW90} that $D(f) = C(f)$, where they then used this correspondence to show a tight $\Omega(\log^2 n)$ depth lower bound for monotone circuits solving undirected $s$-$t$ connectivity. The work of \cite{KarchmerRW95} then proposed a strategy to separate the complexity classes $\mathbf{NC}^1$ and $\mathbf{P}$: start with a function $f$ on $\log n$ bits requiring depth $\Omega(\log n)$, then ``compose'' it with itself $k = \log n / \log\log n$ times (see \cite{KarchmerW90} for a precise definition of composition). If one could prove a strong enough direct sum theorem for communication complexity after composition, even for a random $f$, such a $k$-fold composition would yield a function that is provably in $\mathbf{P}$ (and in fact, even in $\mathbf{NC}^2$), but not in $\mathbf{NC}^1$. Proving such a direct sum theorem is still wide open, and the statement that it is true is known as the ``KRW conjecture''; see for example the recent works \cite{GavinskyMWW14,DinurM16} toward resolving this conjecture. As a toy problem en route to resolving it, \cite{KarchmerRW95} suggested proving a direct sum theorem for $k$-fold composition of a particular function $\ur$ that they defined. That task was positively resolved in \cite{EIRS91} (see also \cite{HastadW90}).

\footnotetext{\label{notejst}\cite{AhnGM12a} writes their algorithm as only needing $\delta$ a constant, but for a different definition of \suppfind{}: when the data structure fails, it should output \textsf{Fail} instead of behaving arbitrarily. They then cite \cite{JowhariST11} as providing the sampler they use, but unfortunately \cite{JowhariST11} does not solve this variant of this problem. This issue can be avoided by using \cite{JowhariST11} with $\delta < 1/\mathop{poly}(n)$ so that whp no failures occur throughout their algorithm.}

The problem $\ur$ abstracts away the function $f$, and Alice and Bob are simply given $x,y\in\{0,1\}^n$ with the promise that $x\neq y$. The players must then agree on any index $i$ with $x_i\neq y_i$. The deterministic communication complexity of $\ur$ is nearly completely understood, with upper and lower bounds that match up to an additive $3$ bits, even if one imposes an upper bound on the number of rounds of communication \cite{TardosZ97}. Henceforth we also consider a generalized problem $\ur_k$, where the output must be $\min\{k, \|x-y\|_0\}$ distinct indices on which $x, y$ differ. We also use $\ur^{\subset}, \ur_k^{\subset}$ to denote the variants when promised $\supp(y)\subset \supp(x)$, and also Bob knows $\|x\|_0$. Clearly $\ur, \ur_k$ can only be harder than $\ur^\subset, \ur_k^\subset$, respectively.

More than twenty years after its initial introduction in connection with circuit depth lower bounds, Jowhari et al.\ in \cite{JowhariST11} demonstrated the relevance of $\ur$ in the randomized one-way communication model for obtaining space lower bounds for certain streaming problems, such as various sampling problems and finding duplicates in streams. In the one-way version, Bob simply needs to find such an index $i$ after a single message from Alice, and we only charge Alice's single message's length as the communication cost. If $\randcom^{\rightarrow,pub}_\delta(f)$ denotes the randomized one-way communication complexity of $f$ in the public coin model with failure probability $\delta$, \cite{JowhariST11} showed that the space complexity of \findup{n} with failure probability $\delta$ is at least $\randcom^{\rightarrow,pub}_{\frac 78 + \frac{\delta}8}(\ur)$. In \findup{n}, one is given a length-$(n+1)$ stream of integers in $[n]$, and the algorithm must output some element $i\in[n]$ which appeared at least twice in the stream (note that at least one such element must exist, by the pigeonhole principle). The work \cite{JowhariST11} then showed a reduction demonstrating that any solution to $\ell_0$-sampling with failure probability $\delta$ in turnstile streams immediately implies a solution to \findup{n} with failure probability at most $(1+\delta)/2$ in the same space (and thus the space must be at least $\randcom^{\rightarrow,pub}_{\frac{15}{16} + \frac{\delta}{16}}(\ur)$). The same result is shown for $\ell_p$-sampling for any $p>0$, in which the output index should equal $i$ with probability $|x_i|^p/(\sum_j |x_j|^p)$, and a similar result is shown even if the distribution on $i$ only has to be close to this $\ell_p$-distribution in variational distance (namely, the distance should be bounded away from $1$). It is then shown in \cite{JowhariST11} that $\randcom^{\rightarrow,pub}_\delta(\ur) = \Omega(\log^2 n)$ for any $\delta$ bounded away from $1$. The approach used though unfortunately does not provide an improved lower bound for $\delta\downarrow 0$.

Seemingly unnoticed in \cite{JowhariST11}, we first point out here that the lower bound proof for $\ur$ in that work actually proves the same lower bound for the promise problem $\ur^\subset$. This observation has several advantages. First, it makes the reductions to the streaming problems trivial (they were already quite simple when reducing from $\ur$, but now they are even simpler). Second, a simple reduction from $\ur^\subset$ to sampling problems provides space lower bounds even in the strict turnstile model, and even for the simpler {\em \suppfind{}} streaming problem for which when queried is allowed to return {\em any} element of $\supp(z)$, without any requirement on the distribution of the index output. Both of these differences are important for the meaningfulness of the lower bound. This is because in dynamic graph streaming applications, typically $z$ is indexed by $\binom{n}{2}$ for some graph on $n$ vertices, and $z_e$ is the number of copies of edge $e$ in some underlying multigraph. Edges then are not deleted unless they had previously been inserted, thus only requiring correctness for strict turnstile streams. Also, for every single application mentioned in the first paragraph of Section~\ref{sec:intro} (except for the two applications in \cite{FrahlingIS08}), the known algorithmic solutions which we cited as using $\ell_0$-sampling as a subroutine actually only need a subroutine for the easier \suppfind{} problem. Finally, third and most relevant to our current work's main focus, the straightforward reductions from $\ur^\subset$ to the streaming problems we are considering here do not suffer any increase in failure probability, allowing us to transfer lower bounds on $\randcom^{\rightarrow,pub}_{\delta}(\ur^\subset)$ for small $\delta$ to lower bounds on various streaming problems for small $\delta$. The work \cite{JowhariST11} could not provide lower bounds for the streaming problems considered there in terms of $\delta$ for small $\delta$.

We now show simple reductions from $\ur^\subset$ to \findup{n} and from $\ur_k^\subset$ to \suppfind{k}. In \suppfind{k} we must report $\min\{k,\|z\|_0\}$ elements in $\supp(z)$. In the claims below, $\delta$ is the failure probability for the considered streaming problem.

\begin{claim}
Any one-pass streaming algorithm for \findup{n} must use $\randcom^{\rightarrow,pub}_{\delta}(\ur^\subset)$ space.
\end{claim}
\begin{proof}
  We reduce from $\ur^\subset$. Suppose there were a space-$S$ algorithm $\mathcal{A}$ for \findup{n}. Alice creates a stream consisting of all elements of $\supp(x)$ and runs $\mathcal{A}$ on those elements, then sends the memory contents of $\mathcal{A}$ to Bob. Bob then continues running $\mathcal{A}$ on $n+1-\|x\|_0$ arbitrarily chosen elements of $[n]\backslash\supp(y)$. Then there must be a duplicate in the resulting concatenated stream, $i$ satisfies $x_i\neq y_i$ iff $i$ is a duplicate.
\end{proof}

\begin{claim}
Any one-pass streaming algorithm for \suppfind{k} in the strict turnstile model must use $\randcom^{\rightarrow,pub}_{\delta}(\ur_k^\subset)$ bits of space, even if promised that $z\in\{0,1\}^n$ at all points in the stream.
\end{claim}
\begin{proof}
This is again via reduction from $\ur_k^\subset$. Let $\mathcal{A}$ be a space-$S$ algorithm for \suppfind{k} in the strict turnstile model. For each $i\in\supp(x)$, Alice sends the update $z_i \leftarrow z_i + 1$ to $\mathcal{A}$. Alice then sends the memory contents of $\mathcal{A}$ to Bob. Bob then for each $i\in\supp(y)$ sends the update $z_i\leftarrow z_i - 1$ to $\mathcal{A}$. Now note that $z$ is exactly the indicator vector of the set $\{i : x_i\neq y_i\}$.
\end{proof}

\begin{claim}
Any one-pass streaming algorithm for $\ell_p$-sampling for any $p\ge 0$ in the strict turnstile model must use $\randcom^{\rightarrow,pub}_{\delta}(\ur_k^\subset)$ bits of space, even if promised $z\in\{0,1\}^n$ at all points in the stream.
\end{claim}
\begin{proof}
This is via straightforward reduction from \suppfind{k}, since reporting $\min\{k,\|z\|_0\}$ elements of $\supp(z)$ satisfying some distributional requirements is only a harder problem than finding {\em any} $\min\{k,\|z\|_0\}$ elements of $\supp(z)$.
\end{proof}

The reductions above thus raise the question: what is the asymptotic behavior of $\randcom^{\rightarrow,pub}_\delta(\ur_k^\subset)$?

\paragraph{Our main contribution:} We prove for any $\delta$ bounded away from $1$ and $k\in[n]$, $\randcom^{\rightarrow,pub}_\delta(\ur_k^\subset) = \Theta(\min\{n, t\log^2(n/t)\})$ where $t = \max\{k,\log(1/\delta)\}$. Given known upper bounds in \cite{JowhariST11}, our lower bounds are optimal for \findup{n}, \suppfind{}, and $\ell_p$-sampling for any $0\le p<2$ for nearly the full range of $n, \delta$ (namely, for $\delta > 2^{-n^{.99}}$). Also given an upper bound of \cite{JowhariST11}, our lower bound is optimal for $\ell_0$-sampling$_k$ for nearly the full range of parameters $n, k, \delta$ (namely, for $t < n^{.99}$). Previously no lower bounds were known in terms of $\delta$ (or $k$). Our main theorem:

\begin{theorem}\label{thm:main}
For any $\delta$ bounded away from $1$ and $1\le k\le n$, $\randcom^{\rightarrow,pub}_\delta(\ur_k^\subset) = \Theta(\min\{n, t\log^2(n/t)\})$.
\end{theorem}

We give two different proofs of Theorem~\ref{thm:main} (in Sections \ref{sec:adaptive-proof} and \ref{sec:aug-proof}). Our upper bound is also new, though follows by minor modifications of the upper bound in \cite{JowhariST11} and thus we describe it in the appendix. The previous upper bound was $O(\min\{n, t\log^2 n\})$. We also mention here that it is known that the upper bound for both $\ur_k$ and $\ell_0$-sampling$_k$ in two rounds (respectively, two passes) is only $O(t\log n)$ \cite{JowhariST11}. Thus, one cannot hope to extend our new lower bound to two or more passes, since it simply is not true.

\subsection{Related work}
The question of whether $\ell_0$-sampling is possible in low memory in turnstile streams was first asked in \cite{CormodeMR05,FrahlingIS08}. The work \cite{FrahlingIS08} applied $\ell_0$-sampling as a subroutine in approximating the cost of the Euclidean minimum spanning tree of a subset $S$ of a discrete geometric space subject to insertions and deletions. The algorithm given there used space $O(\log^3 n)$ bits to achieve failure probability $1/\poly(n)$ (though it is likely that the space could be improved to $O(\log^2 n\log\log n)$ with a worse failure probability, by replacing a subroutine used there with a more recent $\ell_0$-estimation algorithm of \cite{KaneNW10}). As mentioned, the currently best known upper bound solves $\ell_0$-sampling$_k$ using $O(t\log^2 n)$ bits \cite{JowhariST11}, which Theorem~\ref{thm:main} shows is tight.

For $\ell_p$-sampling, conditioned on not failing, the data structure should output $i$ with probability $(1\pm\eps)|x_i|^p/\|x\|_p^p$. The first work to realize its importance came even earlier than for $\ell_0$-sampling: \cite{CoppersmithK04} showed that an $\ell_2$-sampler using small memory would lead to a nearly space-optimal streaming algorithm for multiplicatively estimating $\|x\|_3$ in the turnstile model, but did not know how to implement such a data structure. The first implementation was given in \cite{MonemizadehW10}, achieving space $\poly(\eps^{-1}\log n)$ with $\delta = 1/\poly(n)$. . For $1\le p\le 2$ the space was improved to $O(\eps^{-p}\log^3 n)$ bits for constant $\delta$ \cite{AndoniKO11}. In \cite{JowhariST11} this bound was improved to $O(\eps^{-\max\{1,p\}}\log(1/\delta)\log^2 n)$ bits for failure probability $\delta$ when $0<p<2$ and $p\neq 1$. For $p=1$ the space bound achieved by \cite{JowhariST11} was a $\log(1/\eps)$ factor worse: $O(\eps^{-1}\log(1/\eps)\log(1/\delta)\log^2 n)$ bits.

For finding a duplicate item in a stream, the question of whether a space-efficient randomized algorithm exists was asked in \cite{Muthukrishnan05,Tarui07}. The question was positively resolved in \cite{GopalanR09}, which gave an $O(\log^3 n)$-space algorithm with constant failure probability. An improved algorithm was given in \cite{JowhariST11}, using $O(\log(1/\delta) \log^2 n)$ bits of space for failure probability $\delta$.

\section{Overview of techniques}\label{sec:overview}
We now describe our two proofs of Theorem~\ref{thm:main}. For the upper bound, \cite{JowhariST11} achieved $O(t\log^2n)$, but in the appendix we show that slight modifications to their approach yield $O(\min\{n,t\log^2(n/t)\})$. Our main contribution is in proving an improved lower bound. Assume $t < cn$ for some sufficiently small constant $c$ (since otherwise we already obtain an $\Omega(n)$ lower bound). In both our lower bound proofs in this regime, the proof is split into two parts: we show $\randcom^{\rightarrow,pub}_\delta(\ur^\subset) = \Omega(\log \frac 1{\delta}\log^2 \frac n{\log\frac 1{\delta}})$ and $\randcom^{\rightarrow,pub}_{.99}(\ur_k^\subset)=\Omega(k\log^2\frac nk)$ separately. We give an overview the former here, which is the more technically challenging half. Our two proofs of the latter are in Sections~\ref{sec:k-samples-lb} and \ref{sec:aug-k}.

\subsection{Lower bound proof via encoding subsets and an adaptivity lemma}\label{sec:adaptivity-intro}

Our first proof of the lower bound on $\randcom^{\rightarrow,pub}_\delta(\ur^\subset)$ is via an encoding argument. Fix $m$. A randomized encoder is given a set $S\subset[n]$ with $|S| = m$ and must output an encoding $\enc(S)$, and a decoder sharing public randomness with the encoder must be able to recover $S$ given only $\enc(S)$. We consider such schemes in which the decoder must succeed with probability $1$, and the encoding length is a random variable. Any such encoding must use $\Omega(\log(^n_m)) = \Omega(m\log \frac nm)$ bits in expectation for some $S$.

There is a natural, but sub-optimal approach to using a public-coin one-way protocol $\mathcal{P}$ for $\ur^\subset$ to devise such an encoding/decoding scheme.  The encoder pretends to be Alice with input $x$ being the indicator set of $S$, then lets $\enc(S)$ be the message $M$ Alice would have sent to Bob. The decoder attempts to recover $S$ by iteratively pretending to be Bob $m$ times, initially pretending to have input $y=0\in\{0,1\}^n$, then iteratively adding elements found in $S$ to $y$'s support. Henceforth let $\mathbf{1}_T\in\{0,1\}^n$ denote the indicator vector of a set $T\subset[n]$.

\begin{algorithm}[H] 
  \caption{Simple Decoder.} \label{algo:wrong}
  \begin{algorithmic}[1]
    \Procedure{$\dec$}{$M$}
    \State $T\leftarrow \emptyset$
    \For {$r=1,\ldots,m$} 
      \State Let $i$ be Bob's output upon receiving message $M$ from Alice when Bob's input is $\mathbf{1}_T$
      \State $T \leftarrow T \cup\{i\}$
    \EndFor
    \State \Return $T$
    \EndProcedure
  \end{algorithmic}
\end{algorithm}

One might hope to say that if the original failure probability were $\delta < 1/m$, then by a union bound, with constant probability every iteration succeeds in finding a new element of $S$ (or one could even first apply some error-correction to $x$ so that the decoder could recover $S$ even if only a constant fraction of iterations succeeded). The problem with such thinking though is that this decoder chooses $y$'s adaptively! To be specific, $\mathcal{P}$ being a correct protocol means
\begin{equation}
\forall x,y\in\{0,1\}^n,\ \Pr_s(\mathcal{P}\text{ is correct on inputs }x,y) \ge 1-\delta , \label{eqn:correct}
\end{equation}
where $s$ is the public random string that both Alice and Bob have access to. The issue is that even in the second iteration (when $r=2$), Bob's ``input'' $\mathbf{1}_T$ {\em depends on $s$}, since $T$ depends on the outcome of the first iteration! Thus the guarantee of \eqref{eqn:correct} does not apply.

One way around the above issue is to realize that as long as every iteration succeeds, $T$ is always a subset of $S$. Thus it suffices for the following event $\mathcal{E}$ to occur: $\forall T\subset S,\ \mathcal{P}\text{ is correct on inputs }\mathbf{1}_S, \mathbf{1}_T$. Then $\Pr_s(\neg \mathcal{E}) \le 2^m\delta$ by a union bound, which is at most $1/2$ for $m = \lfloor \log_2(1/\delta)\rfloor - 1$. We have thus just shown that $\randcom^{\rightarrow,pub}_\delta(\ur^\subset) = \Omega(\min\{n, \log(^n_m)\}) = \Omega(\min\{n, \log\frac 1{\delta}\log \frac n{\log(1/\delta)}\})$.

Our improvement is as follows. Our new decoder again iteratively tries to recover elements of $S$ as before. We will give up though on having $m$ iterations and hoping for all (or even most) of them to succeed. Instead, we will only have $R = \Theta(\log \frac 1{\delta}\log \frac n{\log \frac 1{\delta}})$ iterations, and our aim is for the decoder to succeed in finding a new element in $S$ for at least a constant fraction of these $R$ iterations. Simplifying things for a moment, let us pretend for now that all $R$ iterations do succeed in finding a new element. $\enc(S)$ will then be Alice's message $M$, together with the set $B\subset S$ of size $m-R$ not recovered during the $R$ rounds, explicitly written using $\lceil\log{n \choose |B|}\rceil$ bits. If the decoder can then recover these $R$ remaining elements, this then implies the decoder has recovered $S$, and thus we must have $|M| = \Omega(\log{n\choose m} - \log{n \choose |B|}) = \Omega(R\log \frac nm)$. The decoder proceeds as follows. Just as before, initially the decoder starts with $T = \emptyset$ and lets $i$ be the output of Bob on $\mathbf{1}_T$ and adds it to $T$. Then in iteration $r$, before proceeding to the next iteration, the decoder randomly picks some elements from $B$ and adds them into $T$, so that the number of elements left to be uncovered is some fixed number $n_r$. These extra elements being added to $T$ should be viewed as ``random noise'' to mask information about the random string $s$ used by $\mathcal{P}$, an idea very loosely inspired by ideas in differential privacy. For intuition, as an example suppose the iteration $r=1$ succeeds in finding some $i\in S$. If the decoder were then to add $i$ to $T$, as well as $\approx m/2$ random elements from $B$ to $T$, then the resulting $T$ reveals only $\approx 1$ bit of information about $i$ (and hence about $s$). This is as opposed to the $\log m$ bits $T$ could have revealed if the masking were not performed. Thus the next query in round $r=2$, although correlated with $s$, has very weak correlation after masking and we thus might hope for it to succeed. This intuition is captured in the following lemma, which we prove in Section~\ref{sec:optimal-lb}:
\begin{lemma}\label{lem:information}
  Consider $f$: $\{0,1\}^b\times \{0,1\}^q\rightarrow \{0,1\}$ and $X\in\{0,1\}^b$ uniformly random. If $\forall y\in \{0,1\}^q,\ \Pr(f(X,y)=1)\le \delta$ where $0<\delta<1$, then for any random variable $Y$ supported on $\{0,1\}^q$,
  \begin{align}
    \Pr(f(X,Y)=1)\le \frac{I(X;Y)+H_2(\delta)}{\log \frac{1}{\delta}}, \label{eqn:adaptivity}
  \end{align}
  where $I(X;Y)$ is the mutual information between $X$ and $Y$, and $H_2$ is the binary entropy function.
\end{lemma}
Fix some $x\in\{0,1\}^n$. One should imagine here that $f(X,y)$ is $1$ iff $\mathcal{P}$ fails when Alice has input $x$ and Bob has input $y$ in a $\ur^\subset$ instance, and the public random string is $X=s$. Then the lemma states that if $y=Y$ is not arbitrary, but rather random (and correlated with $X$), then the failure probability of the protocol is still bounded as long as the mutual information between $X$ and $Y$ is bounded. It is also not hard to see that this lemma is sharp up to small additive terms. Consider the case $x,y\in[n]$, and $f(x,y) = 1$ iff $x = y$. Then if $X$ is uniform, for all $y$ we have $\Pr(f(X,y) = 1) = 1/n$. Now consider the case where $Y$ is random and equal to $X$ with probability $t/\log n$ and is uniform in $[n]$ with probability $1 - t/\log n$. Then in expectation $Y$ reveals $t$ bits of $X$, so that $I(X;Y) = t$. It is also not hard to see that $\Pr(f(X,Y) = 1) \approx t/\log n + 1/n$.

In light of the strategy stated so far and Lemma~\ref{lem:information}, the path forward is clear: at each iteration $r$, we should add enough random masking elements to $T$ to keep the mutual information between $T$ and all previously added elements below, say, $\frac 12 \log \frac 1{\delta}$. Then we expect a constant fraction of iterations to succeed. The encoder knows which iterations do not succeed since it shares public randomness with the decoder (and can thus simulate it), so it can simply tell the decoder which rounds are the failed ones, then explicitly include in $M$ correct new elements of $S$ for the decoder to use in the place of Bob's wrong output in those rounds. A calculation shows that if one adds a $(1-1/K)\approx 2^{-1/K}$ fraction of the remaining items in $S$ to $T$ after drawing one more support element from Bob, the mutual information between the next query to Bob and the randomness used by $\mathcal{P}$ will be $O(K)$ (see Lemma~\ref{lemma:mutual-entropy-bound}). Thus we do this for $K$ a sufficiently small constant times $\log \frac 1{\delta}$. We will then have $n_r \approx (1 - 1/K)^r m$. Note that we cannot continue in this way once $n_r < K$ (since the number of ``random noise'' elements we inject should at least be one). Thus we are forced to stop after $R = \Theta(K\log(m/K)) = \Theta(\log\frac 1{\delta} \log\frac n{\log \frac 1{\delta}})$ iterations. We then set $m = \sqrt{n\log(1/\delta)}$, so that $\randcom^{\rightarrow,pub}_\delta(\ur^\subset) = \Omega(|R|\log \frac nm) = \Omega(\min\{n, \log\frac 1{\delta}\log^2 \frac n{\log \frac 1{\delta}}\})$ as desired.

The argument for lower bounding $\randcom^{\rightarrow,pub}_\delta(\ur_k^\subset)$ is a bit simpler, and in particular does not need rely on Lemma~\ref{lem:information}. Both the idea and rigorous argument can be found in Section~\ref{sec:k-samples-lb}, but again the idea is to use a protocol for this problem to encode appropriately sized subsets of $[n]$.
 
As mentioned above, our lower bounds use protocols for $\ur^\subset$ and $\ur^\subset_k$ to establish protocols for encoding subsets of some fixed size $m$ of $[n]$. These encoders always consist of some message $M$ Alice would have sent in a $\ur^\subset$ or $\ur^\subset_k$ protocol, together with a random subset $B\subset S$ (using $\lceil \log_2|B|\rceil + \lceil\log{n\choose |B|}\rceil$ bits, to represent both $|B|$ and the set $B$ itself). Here $|B|$ is a random variable. These encoders are thus {\em Las Vegas}: the length of the encoding is a random variable, but the encoder/decoder always succeed in compressing and recovering the subset. The final lower bounds then come from the following simple lemma, which follows from the source coding theorem. 

\begin{lemma} \label{lemma:lb-meta}
  Let $\s$ denote the number of bits used by the $\ur^\subset$ or $\ur^\subset_k$ protocol, and let $\s'$ denote the expected number of bits to represent $B$. Then $(1+\s+\s') \ge \log (^n_m)$. In particular, $s \ge \log(^n_m) - s' - 1$.
\end{lemma}

Section~\ref{sec:optimal-lb} provides our first proof that $\randcom^{\rightarrow,pub}_\delta(\ur^\subset) = \Omega(\min\{n, \log^2(\frac n{\log(1/\delta)}) \log \frac{1}{\delta}\})$. We extend our results in Section~\ref{sec:k-samples-lb} to $\ur_k^\subset$ for $k\ge 1$, proving a lower bound of $\Omega(k\log^2(n/k))$ communication even for constant failure probability.

\subsection{Lower bound proof via reduction from $\aug_N$}

Our second lower bound proof for $\ur^\subset$ is via a randomized reduction from $\aug_N$ \cite{MiltersenNSW98}. In this problem, Charlie receives $z\in\{0,1\}^N$ and Diane receives $j^*\in[N]$ together with $z_j$ for $j=j^*+1,\ldots,N$, and Diane must output $z_{j^*}$. It is shown in \cite{MiltersenNSW98} that $\randcom^{\rightarrow, pub}_\delta(\aug_N) = \Omega(N)$ for any $\delta$ bounded away from $1/2$. In our reduction, $N = \Theta(\log(1/\delta)\log^2\frac n{\log(1/\delta)})$.

For $\ur^\subset$, we can also think of the problem as Alice being given $S\subseteq[n]$ and Bob being given $T\subsetneq S$, and Bob must output some element of $S\backslash T$. In $\aug_N$, Charlie views his input as $L = \Theta(\log \frac n{\log(1/\delta)})$ blocks of bits of nearly equal size, where the $i$th block represents a subset $S_i$ of $[u_i]$ in some collection $\mathcal S_{u_i,m}$ of sets, for some carefully chosen universe sizes $u_i$ per block. Here $\mathcal S_{u_i,m}$ is a collection of subsets of $[u_i]$ of size $m$ of maximal size such any two sets in the collection have intersection size strictly less than $m/2$. Furthermore,  Diane's index $j^*$ is in some particular block of bits corresponding to some set $S_{i^*}$, and Diane also knows $S_i$ for $i>j$.

Now we explain the reduction. We assume some protocol $\mathcal P$ for $\ur^\subset$, and we give a protocol $\mathcal P'$ for $\aug_N$. First, we define the universe $A = \bigcup_{i=1}^L (\{i\} \times [u_i]\times [100^i])$, which has size $n$. Charlie then defines $S = \bigcup_{i=1}^L (\{i\} \times S_i \times [100^i])$. Charlie and
Diane use public randomness to define a uniformly random permutation $\pi$ on $[n]$. Charlie can compute
$\pi(S)$. Also, since Diane knows $S_i$ for $i > i^*$, she can define $T =  \bigcup_{i=i^*+1}^L (\{i\} \times S_i \times [100^i])$
and compute $\pi(T)$. Then $\pi(S)$ and $\pi(T)$ are the inputs to Alice and Bob in the protocol $\mathcal P$
for $\ur^\subset$. Charlie sends Diane the message Alice would have sent Bob in $\mathcal P$ if her input had been $\pi(S)$, and Diane simulates Bob to recover an element in $\pi(S)\backslash \pi(T)$.
Importantly, Alice and Bob do not know anything about $\pi$ at this point other than that $\pi(S) = S$ and $\pi(T) = T$. Thus,
the protocol $\mathcal P$ for $\ur^\subset$, if it succeeds, outputs an arbitrary element $j \in \pi(S) \backslash \pi(T)$, which is a 
deterministic function of the labels of elements in $\pi(S)$ and $\pi(T)$ and the randomness $R$ that Alice and Bob share, which is independent from the randomness in $\pi$. Since $\pi$ is still a uniformly random map conditioned on $\pi(S) = S$
and $\pi(t) = t$ for each $t \in T$, and $j \in \pi(S) \backslash \pi(T)$, it follows that $\pi^{-1}(t)$ is a uniformly random element
of $S \setminus T$. After receiving $\pi^{-1}(j)$, if $(i, a, r) = \pi^{-1}(j)$, then
Charlie and Diane reveal the pairs $((i, a, z), \pi((i,a,z)))$ for each $z \in [100^i]$ to Alice and Bob and Bob updates
his set $\pi(T)$ to include $\pi(i,a,z)$ for each $z \in [100^i]$. One can show that at each step in this process, if Alice and Bob
succeed in outputting an arbitrary item $j$ from $\pi(S) \setminus \pi(T)$, then this is a uniformly random item from
$\pi(S) \setminus \pi(T)$. The fact that this item is uniformly random is crucial for arguing the number of computation paths
of the protocol of Alice and Bob is $o(1/\delta)$ with good probability, over $\pi$, so that one can argue (see below) that with good probability on
every such computation path Alice and Bob succeed on that path, over their randomness $R$. Although the idea of using a random permutation appeared in \cite{JowhariST11} to show that any public coin $\ur$ protocol can be made
into one in which a uniformly random element of $S \backslash T$ is output, here we must use this idea adaptively, slowly revealing
information about $\pi$ and arguing that this property is maintained for each of Bob's successive queries.  

Due to geometrically increasing repetitions of items for increasing $i$, a uniformly random element in $S\backslash T$ is roughly $100$ times more likely to correspond to an item in $S_{i^*}$ than in $S_i$ for $i<i^*$. Thus if Diane simulates Bob to recover a random element in $S\backslash T$, it is most likely to recover an element $j$ of $S_{i^*}$. She can then tell Bob to include $\pi(j)$ and its $100^{i^*}$ redundant copies to $\pi(T)$ and iterate.

There are several obstacles to overcome to make this work. First, iterating means using $\mathcal P$ adaptively, which was the same issue that arose in Section~\ref{sec:adaptivity-intro}. Second, a constant fraction of the time ($1/100$), we expect to obtain an element {\em not} in $S_{i^*}$, but rather from some $S_i$ for $i<i^*$. If this happened too often, then Diane would need to execute many 
queries to recover a sufficiently large number of elements from $S_{i^*}$ in order to solve $\aug_N$. This would then require a union bound over too many possible computation paths, which would not be possible as Alice likely would fail on one
of them (over the choice of $R$). 
However, since the random permutation argument above ensures that at each step we receive
a uniformly random item from the current set $S \setminus T$, if we continue for $m$ iterations, we can argue that with large probability, our sequence of inputs $T$ over the iterations with which Diane invokes Bob's output are all likely to come from a family $\mathcal T$ of size at most $2^{O(m)}$. Here we need to carefully
construct this family to contain a smaller number of sets from levels $i$ for which $i^*-i$
is larger so that the overall number of sets is small. Given this, we can union bound over all such $T$, for total failure probability $\delta |\mathcal T| \ll 1$. Furthermore, we can also argue that after $m$ iterations, it is likely that we have recovered at least $m/2$ of the elements from $S_{i^*}$, which is enough to uniquely identify $S_{i^*}\in \mathcal S_{u_i,m}$ by the limited intersection property of $\mathcal S_{u_i,m}$.

\section{Lower bounds via the adaptivity lemma}\label{sec:adaptive-proof}
\subsection{Communication Lower Bound for $\ur^\subset$} \label{sec:optimal-lb}

Consider a protocol $\mathcal{P}$ for $\ur^\subset$ with failure probability $\delta$, operating in the one-way public coin model. When Alice's input is $x$ and Bob's is $y$, Alice sends $\sketch(x)$ to Bob, and Bob outputs $\query(\sketch(x), y)$, which with probability at least $1-\delta$ is in $\supp(x-y)$. As mentioned in Section~\ref{sec:overview}, we use $\mathcal{P}$ as a subroutine in a scheme for encoding/decoding elements of $\binom{[n]}m$ for $m = \lfloor \sqrt{n\log(1/\delta)}\rfloor$. We assume $\log \frac 1{\delta} \le n/64$, since for larger $n$ we have an $\Omega(n)$ lower bound.

\subsubsection{Encoding/decoding scheme}
We now describe our encoding/decoding scheme $(\enc, \dec)$ for elements in ${[n] \choose m}$, which uses $\mathcal{P}$ in a black-box way. The parameters shared by $\enc$ and $\dec$ are given in Algorithm~\ref{algo:para}.

As discussed in Section~\ref{sec:overview}, on input $S\in {[n] \choose m}$, $\enc$ computes $M \leftarrow \sketch(\mathbf{1}_S)$ as part of its output. Moreover, $\enc$ also outputs a subset $B\subseteq S$ computed as follows. Initially $B=S$ and $S_0=S$. $\enc$ proceeds in $R$ rounds.  In round $r\in[R]$, $\enc$ computes $s_r\leftarrow \query(M, \mathbf{1}_{S\backslash S_{r-1}})$.  Let $b$ denote a binary string of length $R$, where $b_r$ records whether $\query$ succeeds in round $r$.  $\enc$ also outputs $b$.  If $s_r\in S_{r-1}$, i.e. $\query(M, \mathbf{1}_{S\backslash S_{r-1}})$ succeeds, $\enc$ sets $b_r=1$ and removes $s_r$ from $B$ (since the decoder can recover $s_r$ from the $\ur^\subset$-protocol, $\enc$ does not need to include it in $B$); otherwise $\enc$ sets $b_r=0$.  At the end of round $r$, $\enc$ picks a uniformly random set $S_r$ in $\binom{S_{r-1}\backslash \{s_r\}}{n_r}$.  In particular, $\enc$ uses its shared randomness with $\dec$ to generate $S_r$ in such a way that $\enc, \dec$ agree on the sets $S_r$ ($\dec$ will actually iteratively construct $C_r = S\backslash S_r$). We present $\enc$ in Algorithm~\ref{algo:enc}.

The decoding process is symmetric.  Let $C_0=\emptyset$ and $A=\emptyset$.  $\dec$ proceeds in $R$ rounds.  On round $r\in[R]$, $\dec$ obtains $s_r\in S\backslash C_{r-1}$ by invoking $\query(M, \mathbf{1}_{C_{r-1}})$.  By construction of $C_{r-1}$ (to be described later), it is guaranteed that $S_{r-1}=S\backslash C_{r-1}$.  Therefore, $\dec$ recovers exactly the same $s_r$ as $\enc$.  $\dec$ initially assigns $C_r\leftarrow C_{r-1}$.  If $b_r=1$, $\dec$ adds $s_r$ to both $A$ and $C_r$.  At the end of round $r$, $\dec$ inserts many random items from $B$ into $C_r$ so that $C_r=S\backslash S_r$.  $\dec$ can achieve this because of the shared random permutation $\pi$ when constructing $S_r$.  In the end, $\dec$ outputs $B\cup A$.  We present $\dec$ in Algorithm~\ref{algo:dec}.

\begin{algorithm}[H] 
  \caption{Variables shared by encoder $\enc$ and decoder $\dec$.} \label{algo:para}
  \begin{algorithmic}[1] 
    \State $m\leftarrow \lfloor \sqrt{n \log\frac{1}{\delta}} \rfloor$ 
    \State $K\leftarrow \lfloor \frac{1}{16}\log \frac{1}{\delta} \rfloor$
    \State $R\leftarrow \lfloor K\log(m/4K) \rfloor$
    \For {$r = 0, \ldots, R$}
      \State $n_r\leftarrow \lfloor m \cdot 2^{-\frac{r}{K}} \rfloor$ \Comment{$|S_r|=n_r$, and $\forall r\ n_r-n_{r+1}\ge 2$}
    \EndFor
    \State $\pi$ is a random permutation on $[n]$ \Comment{Used to generate $S_r$ and $C_r$}
  \end{algorithmic}
\end{algorithm}

\begin{algorithm}[H] 
  \caption{Encoder $\enc$.} \label{algo:enc}
  \begin{algorithmic}[1]
    \Procedure{$\enc$}{$S$}
    \State $M \leftarrow \sketch(\mathbf{1}_S)$
    \State $A\leftarrow \emptyset$ \Comment{the set $\dec$ recovers just from $M$}
    \State $S_0 \leftarrow S$ \Comment{at end of round $r$, $\dec$ still needs to recover $S_r$}
    \For {$r=1,\ldots,R$}
      \State $s_r\leftarrow \query(M, \mathbf{1}_{S\backslash S_{r-1}})$ \Comment{$s_r\mathbin{\stackrel{\rm ?}{\in}} S_{r-1}$ found in round $r$}
      \State $S_r\leftarrow S_{r-1}$
      \If {$s_r\in S_{r-1}$} \Comment{i.e. if $s_r$ is a valid sample}
        \State $b_r\leftarrow 1$ \Comment{$b\in\{0,1\}^R$ indicating which rounds succeed}
        \State $A\leftarrow A \cup \{s_r\}$
        \State $S_r\leftarrow S_r \backslash \{s_r\}$
      \Else 
        \State $b_r\leftarrow 0$
      \EndIf
      \State Remove $|S_r|-n_r$ elements from $S_r$ with smallest $\pi_a$'s among $a\in S_r$ \Comment{now $|S_r|=n_r$}
    \EndFor
    \State \Return ($M$, $S\backslash A$, $b$) 
    \EndProcedure
  \end{algorithmic}
\end{algorithm}

\begin{algorithm}[H] 
  \caption{Decoder $\dec$.} \label{algo:dec}
  \begin{algorithmic}[1]
    \Procedure{$\dec$}{$M$, $B$, $b$}
    \Statex \ \ \ \ \ $\triangleright$ $M$ is $\sketch(\mathbf{1}_S)$
    \Statex \ \ \ \ \ $\triangleright$ $b\in\{0,1\}^R$ indicates rounds in which Bob succeeds
    \Statex \ \ \ \ \ $\triangleright$ $B$ contains all elements of $S$ that $\dec$ doesn't recover via $M$
    \State $A\leftarrow \emptyset$ \Comment{the subset of $S$ $\dec$ recovers just from $M$}
    \State $C_0 \leftarrow \emptyset$ \Comment{subset of $S$ we have built up so far}
    \For {$r=1,\ldots,R$} \Comment{each iteration tries to recover $1$ element via $M$}
      \State $C_r\leftarrow C_{r-1}$
      \If{$b_r=1$} \Comment{this means Bob succeeds in round $r$}
        \State $s_r\leftarrow \query(M, \mathbf{1}_{C_{r-1}})$ \Comment{Invariant: $C_r=S \backslash S_r$ ($S_r$ is defined in $\enc$)}
        \State $A\leftarrow A \cup \{s_r\}$
        \State $C_r\leftarrow C_r \cup \{s_r\}$
      \EndIf
       \State Insert $m-n_r-|C_r|$ items into $C_r$ with smallest $\pi_a$'s among $a\in B\backslash C_r$
       \Statex \Comment{Random masking ``Differential Privacy'' step. Still $n_r$ elements left to recover.}
    \EndFor
    \State \Return $B\cup A$ 
    \EndProcedure
  \end{algorithmic}
\end{algorithm}

\subsubsection{Analysis}

We have two random objects in our encoding/decoding scheme: (1) the random source used by $\mathcal{P}$, denoted by $X$, and (2) the random permutation $\pi$. These are independent.

First, we can prove that $\dec(\enc(S))=S$.  That is, for any fixing of the randomness in $X$ and $\pi$, $\dec$ will always decode $S$ successfully.  It is because $\enc$ and $\dec$ share $X$ and $\pi$, so that $\dec$ essentially simulates $\enc$.  We formally prove this by induction in Lemma~\ref{lemma:zero-fail-prob}.

Now our goal is to prove that by using the $\ur^\subset$-protocol, the number of bits that $\enc$ saves in expectation over the naive $\lceil\log(^n_m)\rceil$-bit encoding is $\Omega(\log \frac{1}{\delta}\log^2 \frac{n}{\log (1/\delta)} )$ bits.  Intuitively, it is equivalent to prove the number of elements that $\enc$ saves is $\Omega(\log \frac{1}{\delta}\log \frac{n}{\log (1/\delta)} )$.
We formalize this in Lemma~\ref{lemma:bits-saving}. 
Note that $\enc$ also needs to output $b$ (i.e., whether the $\query$ succeeds on $R$ rounds), which takes $R$ bits. 
By our setting of parameters, we can afford the loss of $R$ bits.  Thus it is sufficient to prove $\E|B|=|S|-\Omega(\log \frac{1}{\delta}\log \frac{n}{\log (1/\delta)})$. 

We have $|S|-|B|=\sum_{r=1}^{R}b_r$. 
In Lemma~\ref{lem:information}, we prove the probability that $\query$ fails on round $r$ is upper bounded by $\frac{I(X;S_{r-1})+1}{\log \frac{1}{\delta}}$, where $I(X;S_{r-1})$ is the mutual information between $X$ and $S_{r-1}$. 
Furthermore, we will show in Lemma~\ref{lemma:mutual-entropy-bound} that $I(X;S_{r-1})$ is upper bounded by $O(K)$.
By our setting of parameters, we have $\E b_r=\Omega(1)$ and thus $\E(|S|-|B|)=\Omega(R)=\Omega(\log \frac{1}{\delta}\log \frac{n}{\log (1/\delta)})$.
 
\begin{lemma}\label{lemma:zero-fail-prob}
  $\dec(\enc(S))=S$.
\end{lemma}
\begin{proof}
  We claim that for $r=0,\ldots, R$, $\{S_r, C_r\}$ is a partition of $S$ ($S_r$ is defined in Algorithm~\ref{algo:enc}, and $C_r$ in Algorithm~\ref{algo:dec}). We prove the claim by induction on $r$. Our base case is $r=0$, for which the claim holds since $S_0 = S$, $C_0 = \emptyset$.
  
  Assume the claim holds for $r-1$ ($1\le r \le R$), and we consider round $r$.  On round $r$, by induction $S\backslash S_{r-1}=C_{r-1}$, the index $s_r$ obtained by both \enc and \dec are the same.  Initially $S_r=S_{r-1}$ and $C_r=C_{r-1}$, and so $\{S_r,C_r\}$ is a partition of $S$.  If $s_r$ is a valid sample (i.e. $s_r\in S_{r-1}$), then $b_r=1$, and \enc removes $s_r$ from $S_r$ and in the meanwhile \dec inserts $s_r$ into $C_r$, so that $\{S_r, C_r\}$ remains a partition of $S$. Next, \enc repeats removing the $a$ from $S_r$ with the smallest $\pi_a$ value until $|S_r|=n_r$. Symmetrically, \dec repeats inserting the $a$ into $C_r$ with the smallest $\pi_a$ value among $a\in B\backslash C_r$, until $|C_r|=|S|-n_r$. In the end we have $|S_r|+|C_r|=|S|$, so \enc and \dec execute repetition the same number of times.  Moreover, we can prove that during the same iteration of this repeated insertion, the element removed from $S_r$ is exactly the same element inserted to $C_r$.  This is because in the beginning of a repetition $\{S_r, C_r\}$ is a partition of $S$.  We have $B\backslash C_r\subseteq S\backslash C_r=S_r$. Let $a^*$ denote $a\in S_r$ that minimizes $\pi_a$.  Then $a^*\in B\backslash C_r\subseteq S_r$ (since $a^*$ will be removed from $S_r$, it has no chance to be included in $S$ in \enc, so that $B$ contains $a^*$), and $\pi_{a^*}$ is also the smallest among $\{\pi_a : a\in B\backslash C_r\}$.  Thus both $\enc$ and $\dec$ will take $a^{*}$ (for \enc, to remove from $S_r$, and for \dec, to insert into $C_r$).  Therefore, $\{S_r, C_r\}$ remains a partition of $S$.
  
  Given the fact that $\{S_r, C_r\}$ is a partition of $S$, the $s_r$ are the same in \enc and \dec.  Furthermore, $A=\{s_r : b_r=1,r=1,\ldots, R\}$ are the same in \enc and \dec.  We know $A\subseteq S$.  Since \enc outputs $S\backslash A$, and \dec outputs $(S\backslash A)\cup A$, we have $\dec(\enc(S))=S$.
\end{proof}

\begin{lemma} \label{lemma:bits-saving}
Let $W\in \mathbb{N}$ be a random variable with $W\le m$ and $\E W\le m-d$. Then $\E(\log {n \choose m}-\log {n \choose W})\ge d \log (\frac{n}{m}-1)$.
\end{lemma}

\begin{proof}
  \begin{align*}
  \log {n \choose m}-\log {n \choose W}
  &= \log \frac{n!/(m!(n-m)!)}{n!/(W!(n-W)!)} \\
  &= \sum_{i=1}^{m-W}\log \frac{n-W-i+1}{m-i+1} \\
  &\ge (m-W)\cdot \log \frac{n-W}{m} \\
  &\ge (m-W)\cdot \log \frac{n-m}{m}
  \end{align*}
  
  Taking expectation on both sides, we have $\E(\log {n \choose m}-\log {n \choose W})\ge d \log (\frac{n}{m}-1)$. 
\end{proof}

\noindent \textbf{Lemma~\ref{lem:information} (restated).}
  Consider $f$: $\{0,1\}^b\times \{0,1\}^q\rightarrow \{0,1\}$ and $X\in\{0,1\}^b$ uniformly random. If $\forall y\in \{0,1\}^q,\ \Pr(f(X,y)=1)\le \delta$ where $0<\delta<1$, then for any r.v.\ $Y$ supported on $\{0,1\}^q$,
$$
  \Pr(f(X,Y)=1)\le \frac{I(X;Y)+H_2(\delta)}{\log \frac{1}{\delta}} ,
$$
  where $I(X;Y)$ is the mutual information between $X$ and $Y$, and $H_2$ is the binary entropy function.
\begin{proof}
  It is equivalent to prove 
$$I(X;Y)\ge \E(f(X,Y))\cdot \log\frac{1}{\delta}-H_2(\delta).$$
By definition of mutual entropy $I(X;Y)=H(X)-H(X|Y)$, where $H(X)=b$ and we must show
$$H(X|Y)\le H_2(\delta)+(1-\E(f(X,Y)))\cdot b+\E(f(X,Y))\cdot (b-\log\frac{1}{\delta})=b+H_2(\delta)-\E(f(X,Y))\cdot \log\frac{1}{\delta} .$$
  The upper bound for $H(X|Y)$ is obtained by considering the following one-way communication problem: Alice knows both $X$ and $Y$ while Bob only knows $Y$, and Alice must send a single message to Bob so that Bob can recover $X$. The expected message length in an optimal protocol is exactly $H(X|Y)$.  Thus, any protocol gives an upper bound for $H(X|Y)$, and we simply take the following protocol: Alice prepends a $1$ bit to her message iff $f(X,Y) = 1$ (taking $H_2(\delta)$ bits in expectation). Then if $f(X,Y)=0$, Alice sends $X$ directly (taking $b$ bits). Otherwise, when $f(X,Y)=1$, Alice sends the index of $X$ in $\{x|f(x,Y)=1\}$ (taking $\log (\delta 2^b)=b-\log\frac{1}{\delta}$ bits).  
\end{proof}

\begin{corollary}\label{corollary:sampler-failure}
  Let $X$ denote the random source used by the $\ur^\subset$-protocol with failure probability at most $\delta$. If $S$ is a fixed set and $T\subset S$, $\Pr(\query(\sketch(\mathbf{1}_S), \mathbf{1}_T)\not\in S\backslash T)\le \frac{I(X;T)+H_2(\delta)}{\log\frac{1}{\delta}}$.
\end{corollary}

\begin{lemma}\label{lemma:mutual-entropy-bound}
  $I(X;S_r)\le 6K$, for $r=1,\ldots, R$.
\end{lemma}

\begin{proof}
  Note that $I(X;S_r)=H(S_r)-H(S_r|X)$. Since $|S_r|=n_r$ and $S_r\subseteq S$, $H(S_r)\le \log {m \choose n_r}$. Here is the main idea to lower bound $H(S_r|X)$: By definition of conditional entropy, $H(S_r|X)=\sum_x{p_x\cdot H(S_r|X=x)}$. We fix an arbitrary $x$. If we can prove that for any $T\subseteq S$ where $|T|=n_r$, $\Pr(S_r=T|X=x)\le p$, then by definition of entropy we have $H(S_r|X=x)\ge\log\frac{1}{p}$. 
  
  First we can prove for any fixed $T$,
  
  \begin{align}
    \Pr(S_r=T|X=x)\le \prod_{i=1}^{r}{\frac{{n_{i-1}-n_r-1 \choose n_{i-1}-n_i-1}}{{n_{i-1}-1 \choose n_{i-1}-n_i-1}}}.
  \end{align}
  
  We have $\Pr(S_r=T|X=x)=\Pi_{i=1}^{r}{\Pr(T\subseteq S_i|T\subseteq S_{i-1})}$. 
  On round $i$ ($1\le i \le r$), $\enc$ removes $n_{i-1}-n_i$ elements (at least $n_{i-1}-n_i-1$ of which are chosen all at random) from $S_{i-1}$ to obtain $S_i$. 
  Conditioned on the event that $T\subseteq S_{i-1}$, the probability that $T\subseteq S_i$ is at most ${{n_{i-1}-n_r-1 \choose n_{i-1}-n_i-1}}/{{n_{i-1}-1 \choose n_{i-1}-n_i-1}}$, where the equation achieves when $s_i\in S_{i-1}\backslash T$, and $\enc$ takes a uniformly random subset of $S_{i-1}\backslash \{s_i\}$ of size $n_{i-1}-n_i-1$, so that the subset does not intersect with $T$.
  
  Next we can prove 
  
  \begin{align}
    \prod_{i=1}^{r}{\frac{{n_{i-1}-n_r-1 \choose n_{i-1}-n_i-1}}{{n_{i-1}-1 \choose n_{i-1}-n_i-1}}} \le \frac{2^{6K}}{{m \choose n_r}}. \label{eqn:prod-bound}
  \end{align}
    
  For notational simplicity, let $n^{\underline{k}}$ denote $n\cdot (n-1)\ldots (n-k+1)$. We have 
  \begin{align}
    \prod_{i=1}^{r}{\frac{{n_{i-1}-n_r-1 \choose n_{i-1}-n_i-1}}{{n_{i-1}-1 \choose n_{i-1}-n_i-1}}}
    =\prod_{i=1}^{r}\frac{(n_{i-1}-n_r-1)!n_i!}{(n_{i-1}-1)!(n_i-n_r)!}
    =\prod_{i=1}^{r}\frac{n_i^{\underline{n_r}}}{(n_{i-1}-1)^{\underline{n_r}}}
    =\prod_{i=1}^{r} \left( \frac{n_i^{\underline{n_r}}}{n_{i-1}^{\underline{n_r}}}\cdot \frac{n_{i-1}}{n_{i-1}-n_r} \right).
  \end{align}
  
  By telescoping,
  \begin{align}
    \prod_{i=1}^{r} \frac{n_i^{\underline{n_r}}}{n_{i-1}^{\underline{n_r}}}
    =\frac{n_r^{\underline{n_r}}}{n_0^{\underline{n_r}}}
    =\frac{n_r!(n_0-n_r)!}{n_0!}=\frac{1}{{n_0 \choose n_r}}
    =\frac{1}{{m \choose n_r}}.
    \label{eqn:telescoping}
  \end{align}
  
  Moreover, 
  \begin{align}
    \prod_{i=1}^{r} \frac{n_{i-1}}{n_{i-1}-n_r}
    \le\prod_{i=1}^{r} \frac{1}{1-\frac{m\cdot 2^{-r/K}}{m\cdot 2^{-(i-1)/K}-1}}
    \le\prod_{i=1}^{r} \frac{1}{1-\frac{m\cdot 2^{-r/K}+1}{m\cdot 2^{-(i-1)/K}}}
    =\prod_{j=1}^{r} \frac{1}{1-2^{-j/K}-\frac{2^{\frac{r-j}{K}}}m}.
    \label{eqn:bound-with-floor}
  \end{align}
  
  By our setting of parameters 
  $$\frac{2^{\frac rK}}m \le \frac{2^{\frac RK}}m \le \frac{1}{4K} .$$
  
  Therefore, for $j\in \{1,\ldots, r\}$,
  $$\frac{1}{1-2^{-\frac jK}-\frac{2^{\frac{r-j}{K}}}m}\le \frac{1}{1-(1+\frac{1}{4K})2^{-\frac jK}}.$$ 
  
  By Taylor series $2^{1/K} = \sum_{n=0}^{\infty}{\frac{(\ln 2 )^n}{n!K^n}} >1+\frac{\ln 2}{K}>1+\frac{1}{4K}$, and thus $\frac{1}{1-(1+\frac{1}{4K})2^{-j/K}}\le \frac{1}{1-2^{(1-j)/K}}$, for $j=2,\ldots, r$. For $j=1$, we have $\frac{1}{1-(1+\frac{1}{4K})2^{-\frac 1K}} \le 2^K$.
  
  By Lemma~\ref{lemma:Pochhammer}, we have $\prod_{j=1}^{\infty} \frac{1}{1-2^{-j/K}}\le 2^{5K}$. Therefore, the right hand side of \eqref{eqn:bound-with-floor} is upper bounded by $2^{6K}$. Together with \eqref{eqn:telescoping}, we prove \eqref{eqn:prod-bound} holds.  
  
  Finally, let $p={2^{6K}}/{{m\choose n_r}}$, we have $\Pr(S_r=T|X=x)\le p$ and thus $H(S_r|X=x)\ge \log\frac{1}{p}=\log{{m\choose n_r}}-6K$. Therefore, $H(S_r|X)\ge \log{{m\choose n_r}}-6K$ and so $I(X;S_r)=H(S_r)-H(S_r|X)\le 6K$.  
\end{proof}

\begin{lemma}\label{lemma:Pochhammer}
  Let $K\in \mathbb{N}$ and $K\ge 1$. We have $\prod_{j=1}^{\infty} \frac{1}{1-2^{-j/K}}\le 2^{5K}$.
\end{lemma}

\begin{proof}
  First, we bound the product of first $2K$ terms. Note that $\frac{1}{1-2^{-x}}\le \frac{8}{3x}$ for $0<x\le 2$. Therefore, 
  \begin{align}
    \prod_{j=1}^{2K}\frac{1}{1-2^{-j/K}}
    \le (8/3)^{2K}\cdot \frac{K^{2K}}{(2K)!}
    \le (8/3)^{2K}\cdot \frac{K^{2K}}{(2K/e)^{2K}}
    = (4e/3)^{2K}
    < 2^{4K}. 
  \end{align}
  
  Then, we bound the product of the rest terms
  \begin{align}
    \prod_{j=2K+1}^{\infty}\frac{1}{1-2^{-j/K}} 
    \le \prod_{j=2K+1}^{\infty}\frac{1}{1-2^{-\lfloor j/K \rfloor}} 
    \le \prod_{i=2}^{\infty}\left( \frac{1}{1-2^{-i}}\right)^K 
    \le \left( \frac{1}{1-\sum_{i=2}^{\infty}2^{-i}}\right)^K
    = 2^K.
  \end{align}
  
  Multiplying two parts proves the lemma.
\end{proof}

\begin{theorem}
  $\randcom^{\rightarrow,pub}_\delta(\ur^\subset) = \Omega(\log \frac{1}{\delta}\log^2 \frac{n}{\log (1/\delta)} )$, given that $64 \le \log \frac{1}{\delta} \le \frac{n}{64}$.
\end{theorem}

\begin{proof}
  By Lemma~\ref{lemma:zero-fail-prob}, the success probability of protocol $(\enc,\dec)$ is $1$. 
  By Lemma~\ref{lemma:lb-meta}, we have $\s\ge \log (^n_m) - \s' -1$, where $\s'=\log n + R+ \E(\log (^n_{|B|}))$. 
  The size of $B$ is $|B|=|S|-\sum_{r=1}^{R}{b_r}$.
  By Corollary~\ref{corollary:sampler-failure}, conditioned on $S$, $\Pr(b_r=0)\le \frac{I(X;S_{r-1})+1}{\log\frac{1}{\delta}}$. 
  By Lemma~\ref{lemma:mutual-entropy-bound}, $I(X;S_{r-1})\le 6K$ (Note that when $r=1$, $I(X;S_0)=0\le 6K$). 
  Therefore, $\E(b_r)\ge 1-\frac{6K+1}{\log\frac{1}{\delta}}$.
  By the setting of parameters (see Algorithm~\ref{algo:para}) we have $\E(b_r)\ge \frac{39}{64}$. Therefore, $\E(|B|)\le |S|-\frac{39}{64}R$. 
  By Lemma~\ref{lemma:bits-saving}, $\log (^n_m)-\E(\log (^n_{|B|}))\ge \frac{39}{64}R\cdot \log (\frac{n}{m}-1) \ge \frac{1}{2}R\log (\frac{n}{\log(1/\delta)})$. 
  Furthermore, $\frac{1}{6}R\log \frac{n}{\log (1/\delta)} \ge R$.
  Thus we obtain $\s \ge \frac{R}{3}\log \frac{n}{\log(1/\delta)} -(\log n + 1)  =\Omega(\log \frac{1}{\delta}\log^2 \frac{n}{\log (1/\delta)} )$.
\end{proof}

\subsection{Communication Lower Bound for $\ur_k^\subset$}\label{sec:k-samples-lb}

In this section, we prove the lower bound $\randcom^{\rightarrow,pub}_{1/2}(\ur^\subset_k) = \Omega(\min\{n, k\log^2 \frac{n}{k}\})$. In fact, our lower bound holds for any failure probability $\delta$ bounded away from $1$. Let $\mathcal{P}$ denote a $\ur_k^\subset$-protocol where Alice sends $\sketch_k(x)$ to Bob, and Bob outputs $\query_k(\sketch_k(x), y)$.  We consider the following encoding/decoding scheme $(\enc_k, \dec_k)$ for $S\in {[n] \choose m}$.  $\enc_k$ computes $M\leftarrow \sketch_k(\mathbf{1}_S)$ as part of its message. In addition, $\enc_k$ includes $B\subseteq S$ constructed as follows, spending $\lceil\log{n\choose |B|}\rceil$ bits.  Initially $B= S$, and $\enc_k$ proceeds in $R=\Theta(\log (n/k))$ rounds.  Let $S_0=S\supseteq S_1\supseteq \ldots \supseteq S_R$ where $S_r$ is generated by sub-sampling each element in $S_{r-1}$ with probability $\frac{1}{2}$.  In round $r$ ($r=1,\ldots, R$), $\enc_k$ tries to obtain $k$ elements from $S_{r-1}$ by invoking $\query_k(M, \mathbf{1}_{S\backslash S_{r-1}})$, denoted by $A_k$, and removes $A_k\cap (S_{r-1}\backslash S_{r})$ (whose expected size is $\frac{k}{2}$) from $B$.  Note that $\dec_k$ is able to recover the elements in $A_k\cap (S_{r-1}\backslash S_{r})$.  For each round the failure probability of $\query_k$ is at most $\delta$.  Thus we have $\E(|S|-|B|)\ge \frac{k}{2}\cdot (1-\delta) \cdot R=\Omega(k\log\frac{n}{k})$.  Furthermore, each element contains $\Theta(\log \frac{n}{k})$ bits of information, thus yielding a lower bound of $\Omega(k\log^2\frac{n}{k})$ bits.

In this section we assume $k \le n/2^{10}$, since for larger $n$ we have an $\Omega(n)$ lower bound.

\subsubsection{Encoding/decoding scheme}
\begin{algorithm}[H] 
  \caption{Variables Shared by Encoder $\enc_k$ and Decoder $\dec_k$.} \label{algo:para4}
  \begin{algorithmic}[1] 
    \State $m\leftarrow \lfloor \sqrt{nk} \rfloor$
    \State $R\leftarrow \lfloor \frac{1}{2}\log (n/k) - 2 \rfloor$ \Comment{Note that $R\ge 3$ because $k\le \frac{n}{2^{10}}$}
    \State $T_0\leftarrow [n]$
    \For {$r = 1, \ldots, R$}
      \State $T_r\leftarrow \emptyset$
      \State For each $a\in T_{r-1}$, $T_r\leftarrow T_r\cup \{a\}$ with probability $\frac{1}{2}$ \Comment{We have $S_r=S\cap T_r$}
    \EndFor
  \end{algorithmic}
\end{algorithm}

\begin{algorithm}[H] 
  \caption{Encoder $\enc_k$.} \label{algo:enc4}
  \begin{algorithmic}[1]
    \Procedure{$\enc_k$}{$S$}
    \State $M \leftarrow \sketch_k(\mathbf{1}_S)$
    \State $A\leftarrow \emptyset$
    \For {$r=1,\ldots,R$}
    \State $A_r\leftarrow \query_k(M, \mathbf{1}_{S\backslash (S\cap T_{r-1})})$
    \If {$A_r\subseteq S\cap T_{r-1}$} \Comment{i.e. if $A_r$ is valid}
      \State $b_r\leftarrow 1$ \Comment{$b$ is a binary string of length $R$, indicating if $\query_k$ succeeds in round $r$}
      \State $A\leftarrow A \cup (A_r\cap (T_{r-1}\backslash T_r))$
    \Else 
      \State $b_r\leftarrow 0$
    \EndIf
    \EndFor
      \State \Return ($M$, $S\backslash A$, $b$) 
    \EndProcedure
  \end{algorithmic}
\end{algorithm}

\begin{algorithm}[H] 
  \caption{Decoder $\dec_k$.} \label{algo:dec4}
  \begin{algorithmic}[1]
    \Procedure{$\dec_k$}{$M$, $B$, $b$}
    \State $A\leftarrow \emptyset$
    \State $C_0 \leftarrow \emptyset$
    \For {$r=1,\ldots,R$}
      \State $C_r\leftarrow C_{r-1}$
      \If {$b_r=1$}
        \State $A_r\leftarrow \query_k(M, \mathbf{1}_{C_{r-1}})$ \Comment{Invariant: $C_r=S\backslash (S\cap T_r)$}
        \State $A\leftarrow A \cup (A_r\cap (T_{r-1}\backslash T_r))$
        \State $C_r\leftarrow C_r \cup (A_r\cap (T_{r-1}\backslash T_r))$
      \EndIf
      \State $C_r\leftarrow C_r \cup (B\cap (T_{r-1}\backslash T_r))$
    \EndFor
    \State \Return $B\cup A$ 
    \EndProcedure
  \end{algorithmic}
\end{algorithm}

\subsubsection{Analysis}

\begin{theorem}\label{thm:urk}
  $\randcom^{\rightarrow,pub}_\delta(\ur_k^\subset) = \Omega((1-\delta)k\log^2 \frac{n}{k} )$, given that $1 \le k \le \frac{n}{2^{10}}$ and $0<\delta \le 1-\frac{50\log n}{k\log^2(n/k)}$.
\end{theorem}
\begin{proof}
Let $S_r=S\cap T_r$.  Let $\success$ denote the event that $|S\cap T_R|=|S_R|\ge k$.  Note that $\E|S_R|=\frac{1}{2^R}m=4k$. By the Chernoff bound, $\Pr(\success)\ge \frac{1}{2}$.  In the following, we argue conditioned on $\success$. Namely, in each round $r$, there are at least $k$ items in $S_r$.
  
Similar to Lemma~\ref{lemma:zero-fail-prob}, we can prove the protocol $(\enc_k,\dec_k)$ always succeeds.  By Lemma~\ref{lemma:lb-meta}, we have $\s\ge \log (^n_m) - \s' -2$, where $\s'=\log n + R+ \E \log (^n_{|B|})$.  The size of $B$ is $|B|=|S|-\sum_{r=1}^{R}{(b_r \cdot |A_r \cap (S_{r-1}\backslash S_r)|)}$.  The randomness used by $\mathcal{P}$ is independent from $S\backslash S_{r-1}$ for every $r\in[R]$.  Therefore, $\E b_r\ge 1-\delta$, and $b_r$ is independent from $|A_r \cap (S_{r-1}\backslash S_r)|$.  We have $\E|A_r \cap (S_{r-1}\backslash S_r)|=\frac{k}{2}$, and thus $\E(|S|-|B|)\ge \frac{(1-\delta)kR}{2}$.  By Lemma~\ref{lemma:bits-saving}, $\log (^n_m)-\E\log (^n_{|B|})\ge \frac{(1-\delta)kR}{2}\cdot \log (\frac{n}{m}-1) \ge \frac{(1-\delta)kR}{5}\log (\frac{n}{k})$.  Moreover, $R\le \log n$ and $\log n \le \frac{(1-\delta)kR}{12}\log \frac{n}{k}$.  Thus we have $\s = \Omega((1-\delta)kR\log\frac{n}{k}) = \Omega((1-\delta)k\log^2 \frac{n}{k} )$.
\end{proof}

\section{Lower bounds proofs via augmented indexing}\label{sec:aug-proof}

Here we show another route to proving $\randcom^{\rightarrow,pub}_\delta(\ur_k^\subset) = \Omega(\min\{n, t\log^2(n/t)\}$ via reduction from augmented indexing. We again separately prove lower bounds for $\randcom^{\rightarrow,pub}_\delta(\ur^\subset)$ and $\randcom^{\rightarrow,pub}_{\frac 15}(\ur_k^\subset)$. Both proofs make use of the following standard lemma. The proof can be found in the appendix (see Section~\ref{sec:code}).

\begin{lemma}\label{lem:code}
For any integers $u\ge 1$ and $1\le m\le u/(4e)$, there exists a collection $\mathcal S_{u,m} \subset \binom{[u]}m$ with $\log |\mathcal{S}_{u,m}| = \Theta(m\log(u/m))$ such that for all $S\neq S'\in \mathcal S_{u,m}$, $|S\cap S'| < m/2$.
\end{lemma}

Both our lower bounds in Sections~\ref{sec:aug-delta} and \ref{sec:aug-k} reduce from augmented indexing (henceforth $\aug$) to either $\ur^\subset$ with low failure probability, or $\ur_k^\subset$ with constant failure probability, in the public coin one-way model of communication. We remind the reader of the setup for the $\aug_N$ problem. There are two players, Charlie and Diane. Charlie receives $z\in\{0,1\}^N$ and Diane receives $j^*\in[N]$ together with $z_{j^*+1},\ldots,z_N$. Charlie must send a single message to Diane such that Diane can then output $z_{j^*}$. The following theorem is known.

\begin{theorem}{\cite{MiltersenNSW98}}\label{thm:mnsw}
$\randcom^{\rightarrow,pub}_{1/3}(\aug_N) = \Theta(N)$.
\end{theorem}

We show that if there is an $s$-bit communication protocol $\mathcal P$ for $\ur^\subset$ on $n$-bit vectors with failure probability $\delta$ (or for $\ur_k$ with constant failure probability), that implies the existence of an $s$-bit protocol $\mathcal P'$ for $\aug_N$ for some $N=\Theta(\log\frac 1{\delta}\log^2\frac n{\log\frac 1{\delta}})$ (or $N=\Theta(k\log^2(n/k))$ for $\ur_k$). The lower bound on $s$ then follows from Theorem~\ref{thm:mnsw}.

\subsection{Communication Lower Bound for $\ur^\subset$}\label{sec:aug-delta}

Set $t = \log \frac 1{\delta}$. In this section we assume $t < n/(4e)$ and show $\randcom^{\rightarrow,pub}_\delta(\ur^\subset) = \Omega(t\log^2(n/t))$. This implies a lower bound of $\Omega(\min\{n, t\log^2(n/t)\})$ for all $\delta>0$ bounded away from $1$.

As mentioned, we assume we have an $s$-bit protocol $\mathcal P$ for $\ur^\subset$ with failure probability $\delta$, with players Alice and Bob.We use $\mathcal P$ to give an $s$-bit protocol $\mathcal P'$ for $\aug_N$, which has players Charlie and Diane, for $N = \Theta(t\log^2(n/t))$.

The protocol $\mathcal P'$ operates as follows. Without loss of generality we may assume that, using the notation of Lemma~\ref{lem:code}, $|\mathcal S_{u,m}|$ is a power of $2$ for $u, m$ as in the lemma statement. This is accomplished by simply rounding $|\mathcal S_{u,m}|$ down to the nearest power of $2$ by removing elements arbitrarily. Also, define $L = c\log(n/t)$ for some sufficiently small constant $c\in(0,1)$ to be determined later. Now, Charlie partitions the bits of his input $z\in\{0,1\}^N$ into $L$ consecutive sequences of bits such that the $i$th chunk of bits for each $i\in[L]$ can be viewed as specifying an element $S_i\in \mathcal S_{u_i,m}$ for $u_i = \frac n{100^i\cdot L}$ and $m = ct$. Lemma~\ref{lem:code} gives $\log|\mathcal S_{u_i,m}| = \Theta(m\log(u_i/m))$, which is $\Theta(t\log(n/t))$ for $c < 1/14$. Thus $N = \Theta(L\cdot t\log(n/t)) = \Theta(t\log^2(n/t))$. Given these sets $S_1,\ldots,S_L$, we now discuss how Charlie generates a vector $x\in\{0,1\}^n$. Charlie then simulates Alice on $x$ to generate the message Alice would have send to Bob in protocol $\mathcal P$, then sends that same message to Diane.

To generate $x\in\{0,1\}^n$, assume Charlie and Diane have sampled a bijection from 
\begin{equation}\label{eq:pi-origin}
A = \bigcup_{i=1}^L (\{i\} \times [u_i]\times [100^i])
\end{equation}
to $[n]$ uniformly at random. We denote this bijection by $\pi$. This is possible since $|A| = n$. Then Charlie defines $x$ to be the indicator vector $\mathbf{1}_{\pi(S)}$, where
$$
S = \bigcup_{i=1}^L (\{i\} \times S_i \times [100^i]),
$$
then sends a message $M$ to Diane, equal to Alice's message with input $\mathbf{1}_{\pi(S)}$. This completes the description of Charlie's behavior in the protocol $\mathcal P'$.

We describe how Diane uses $M$ to solve $\aug_N$. Diane's input $j^*\in[N]$ lies in some chunk $i^*\in[L]$. We now show how Diane can use $\mathcal P$ to recover $S_{i^*}$ with probability $2/3$ (and thus in particular recover $z_{j^*}$). Since Diane knows $z_j$ for $j>j^*$, she knows $S_i$ for $i>i^*$. She can then execute the following algorithm.

\begin{algorithm}[H] 
  \caption{Behavior of Diane in $\mathcal P'$ for $\ur^\subset$.} \label{algo:diane1}
  \begin{algorithmic}[1]
    \Procedure{$\diane$}{$M$}
    \State $T \leftarrow \bigcup_{i=i^*+1}^L (\{i\} \times S_i \times [100^i])$
    \State $T_{i^*}\leftarrow \emptyset$
    \While {$|T_{i^*}| < \frac m2$}
      \State $(i,a,r)\leftarrow \pi^{-1}(\query(M, \mathbf{1}_{\pi(T)}))$
      \State $T\leftarrow T \cup ((i,a) \times [100^i])$
      \If {$i=i^*$} 
        \State $T_{i^*} \leftarrow T_{i^*}\cup \{a\}$
      \EndIf
    \EndWhile
    \If {there exists $S\in\mathcal S_{u_{i^*},m}$ with $T_{i^*}\subset S$}
      \State \Return the unique such $S$
    \Else
      \State \Return \textsf{Fail}
    \EndIf
    \EndProcedure
  \end{algorithmic}
\end{algorithm}

In Algorithm~\ref{algo:diane1} Diane is building up a subset $T_{i^*}$ of $S_{i^*}$. Once $|T_{i^*}| \ge |S_{i^*}|/2 = m/2$, Diane can uniquely recover $S_{i^*}$ by the limited intersection property of $\mathcal{S}_{u_i,m}$ guaranteed by Lemma~\ref{lem:code}. Until then, she uses $\mathcal P$ to recover elements of $S\backslash T$, which, as we now show, are chosen uniformly at random from $S\setminus T$. 

\begin{claim}\label{cl:uniform}
For every protocol for Alice and Bob that uses shared randomness with Bob's behaviour given by $\query(\cdot)$, for every choice of shared random string $R$ of Alice and Bob, for every $S, T\subset S$, the following conditions hold. If $\pi$ is a uniformly random permutation, the success or failure of $\query(M, \mathbf{1}_{\pi(T)})$ is determined by $\{\pi(j)\}_{j\in T}$ and the image $\pi(S\setminus T)$ of $S\setminus T$ under $\pi$. Conditioned on a choice of $R$, $\{\pi(j)\}_{j\in T}$ and $\pi(S\setminus T)$ such that $\query(M, \mathbf{1}_{\pi(T)})$ succeeds, one has that $\pi^{-1}(\query(M, \mathbf{1}_{\pi(T)}))$ is a uniformly random element of $S\setminus T$.
\end{claim}
\begin{proof}
The first claim follows by noting that the message $M$ that Alice sends to Bob is solely a function of $R$ and $\pi(S)$. The behaviour of Bob is determined by $M$ and $\pi(T)$ (and the latter is determined by $\{\pi(j)\}_{j\in T}$).

Now condition on the values of $R$, $\{\pi(j)\}_{j\in T}$ and $\pi(S\setminus T)$ such that  $\query(M, \mathbf{1}_{\pi(T)})$ succeeds, and let $j^*\in [n]$ denote the output. Note that by our conditioning $j^*$ is a fixed quantity. The only randomness left is the exact mapping of $S\setminus T$ to $\pi(S\setminus T)$. This mapping is independent of $\{\pi(j)\}_{j\in T}$ and $\pi(S\setminus T)$ and uniformly random, so $\pi^{-1}(j^*)$ is a uniformly random element of $S\setminus T$, as required.
\end{proof}

Fix any protocol $\widetilde{\query}(M, \mathbf{1}_{\pi(T)})$ (not necessarily the one that Charlie and Diane use; see analysis of the idealized process $\widetilde{\mathcal{P}}$ below). Now fix $T$ together with values of $R$, $\{\pi(j)\}_{j\in T}$ and $\pi(S\setminus T)$ such that  $\widetilde{\query}(M, \mathbf{1}_{\pi(T)})$ succeeds.  

\if 0
\paragraph{Elements in $S_{i^*}$ are likely to be recovered.}   Given Claim~\ref{cl:uniform}, since the elements of $S_i$ appear with frequency $100^i$ in $S\backslash T$, they are more likely to be returned by $\pi^{-1}(\widetilde{\query}(M, \mathbf{1}_{\pi(T)}))$. Indeed, as long as $T_{i^*}\le m/2$, at least $m/2$ elements remain in $S_{i^*}\backslash T_{i^*}$, implying the output of $(i, a, r)$ of $\pi^{-1}(\widetilde{\query}(M, \mathbf{1}_{\pi(T)}))$ satisfies
\begin{equation}
\begin{split}
\Pr(i = i^* | (R, \{\pi(j)\}_{j\in T}, \pi(S\setminus T)) \text{~s.t.~}\widetilde{\query}(M, \mathbf{1}_{\pi(T)})\text{~succeeds}) &\ge \frac{\frac m2\cdot 100^{i^*}}{\frac m2\cdot 100^{i^*} + m\cdot \sum_{i=1}^{i^*-1} 100^i}\\
& = \frac{\frac 12\cdot 100^{i^*}}{\frac 12\cdot 100^{i^*} + \frac {100}{99}(100^{i^*-1} - 1)}\\
& > \frac{49}{50} , \label{eqn:istar-likely}
\end{split}
\end{equation}
where the probability is over the choice of $\pi|_{S\setminus T}:(S\setminus T)\to \pi(S\setminus T)$ (recall that we condition on the image $\pi(S\setminus T)$ under $\pi$, but not on the actual mapping).
\fi

\paragraph{Elements in $S_{j}, j<i^*,$ are unlikely to be recovered.} Given Claim~\ref{cl:uniform}, since the elements of $S_j$ appear with frequency $100^j$ in $S\backslash T$, they are less likely to be returned by $\pi^{-1}(\widetilde{\query}(M, \mathbf{1}_{\pi(T)}))$ when $j$ is small.  More precisely, as long as $|S_{i^*}\cap T_{i^*}|\geq m/2$, for any $1\le j< i^*$
\begin{equation}
\begin{split}
\Pr(i = j | (R, \{\pi(j)\}_{j\in T}, \pi(S\setminus T)) \text{~s.t.~}\widetilde{\query}(M, \mathbf{1}_{\pi(T)})\text{~succeeds})& \le \frac{m\cdot 100^j}{\frac m2\cdot 100^{i^*}}\\
& \le 2\cdot 100^{-(i^*-j)} \\
&\le 50^{-(i^*-j)} . \label{eqn:others-unlikely}
\end{split}
\end{equation}
Here again the probability is over the choice of $\pi|_{S\setminus T}:(S\setminus T)\to \pi(S\setminus T)$ (recall that we condition on the image $\pi(S\setminus T)$ under $\pi$, but not on the actual mapping).

We now define the set $\mathcal{T}$ of {\bf typical intermediate sets}, which plays a crucial role in our analysis.  Let $Q_i$ for $i\in [L]$ denote $\{i\} \times S_i \times [100^i]$. Let $\mathcal{T}$ be the collection of all $T\subset S$ such that (1) $Q_i\subset T$ for all $i>i^*$, and (2) for each $i < i^*$, $|T\cap Q_i| \le 100^i\cdot  m/4^{i^*-i}$.  The following claim will be useful:
\begin{claim}\label{cl:size-of-t}
For the set $\mathcal{T}$ defined above one has $|\mathcal{T}|=2^{O(m)}$.
\end{claim}
\begin{proof}
\begin{align*}
|\mathcal T| &\le 2^m \cdot \prod_{i=1}^{i^*-1}\left(\sum_{r=0}^{\frac m{4^{i^*-i}}} \binom mr\right)\text{ (the }2^m\text{ term comes from }S_{i^*}\text{)}\\
{}&\le 2^m \cdot \prod_{i=1}^{i^*-1} \binom{m + \frac m{4^{i^* - i}}}{\frac m{2^{i^* - i}}}\\
{}&\le 2^m \cdot \prod_{i=1}^{i^*-1} (2e\cdot 4^{i^*-i})^{\frac m{4^{i^* - i}}}\text{ (using }\binom nk \le (en/k)^k\textrm{)}\\
{}&\le 2^{O(m)} \cdot 2^{m\cdot O(\sum_{j=1}^\infty j 4^{-j})} \\
{}& \le 2^{O(m)}
\end{align*}
\end{proof}

We will show that for most choices of $\pi$ and shared random string $R$ Algorithm~\ref{algo:diane1} {\bf (a)} never leaves the set $\mathcal{T}$ and {\bf (b)} successfully terminates.  Note that Algorithm~\ref{algo:diane1} is a random process whose sample space is the product of the set of all possible permutations $\pi$ and shared random strings $R$. As before, we denote this process by $\mathcal{P}'$. It is useful for analysis purposes to define another process $\widetilde{\mathcal{P}}$, which is an idealized version of $\mathcal{P}'$. In this process instead of running  $\query(M, \mathbf{1}_{\pi(T)})$ Alice runs  $\widetilde{\query}(M, \mathbf{1}_{\pi(T)})$, which is guaranteed to output an element of $\pi(S\setminus T)$ for every choice of  $T\subset S$, shared random string $R$, $\{\pi(j)\}_{j\in T}$, and $\pi(S\setminus T)$. The proof proceeds in three steps.

{\bf Step 1: proving that $\widetilde{\mathcal{P}}$ succeeds in recovering $T_{i^*}$ and never leaves $\mathcal{T}$ with high probability.}
Choose $\pi$ uniformly at random. By \eqref{eqn:others-unlikely}, as long as $|S_{i^*}\cap T_{i^*}|\geq m/2$, the expected number of items recovered by $\widetilde{\query}$ from $S_i$ for $i<i^*$ in the first $m$ iterations is at most $m/50^{i^*-i}$. Thus the probability of recovering more than $m/4^{i^*-i}$ items from $S_i$ is at most $(1/12)^{i^*-i}$ by Markov's inequality. Note that the probability is over the choice of $\pi$ only, as $\widetilde{\query}$ is assumed to succeed with probability $1$ by definition of $\widetilde{\mathcal{P}}$. 
Thus
$$
\Pr( \widetilde{\mathcal{P}}\text{~leaves~}\mathcal{T}) \le \sum_{i=1}^{i^*-1}\left(1/12\right)^{i^*-i} < 1/10.
$$
In particular this means that with probability at least $1-1/10$ at most $\sum_{i<i^*} m/4^{i^*-i}<m/2$ items from $\bigcup_{i<i^*} S_i$ are recovered in the first $m$ (or fewer, if the algorithm terminates earlier) iterations. This also implies that with probability at least $1-1/10$ if the algorithm proceeds for the entire $m$ iterations, it recovers at least $m/2$ elements of $T_{i^*}$ and hence terminates. We thus get that $\widetilde{\mathcal{P}}$ succeeds at least with probability $1-1/10$.

{\bf Step 2: coupling $\widetilde{\mathcal{P}}$ to $\mathcal{P}'$ on most of the probability space.}
For every $T\subset S$ and every $\pi$ let $\mathcal{E}_T(\pi)$ be the probabilistic event (over the choice of $\query$'s random string $R$) that $\query(M, \mathbf{1}_{\pi(T)})$ succeeds in returning an element in $\pi(S\backslash T)$. Note that $\mathcal{E}_T(\pi)$ is a subset of the probability space of shared random strings $R$, and depends on $\pi$. We let 
$$
\mathcal{E}_{\mathcal T}(\pi):=\wedge_{T\in\mathcal T} \mathcal E_T(\pi)
$$
to simplify notation. Using Claim~\ref{cl:size-of-t} and the union bound we have for every $\pi$
$$
\Pr_R(\neg(\mathcal E_{\mathcal T}(\pi))) \le \delta\cdot |\mathcal T|\leq 1/20
$$
as long as for $m = c\log(1/\delta)$ for $c$ a sufficiently small constant.

Now recall that $\widetilde{\query}(M, \mathbf{1}_{\pi(T)})$ is an idealized protocol, which is guaranteed to output an element of $\pi(S\setminus T)$ for every choice of  $T\subset S$, shared random string $R$, $\{\pi(j)\}_{j\in T}$, and $\pi(S\setminus T)$. We have just shown that for every $\pi$ the event ${\mathcal E_{\mathcal T}(\pi)}$ occurs  with probability at least $1-1/20$ over the choice of $R$. Now define $\widetilde{\query}(M, \mathbf{1}_{\pi(T)})$ as equal to $\query(M, \mathbf{1}_{\pi(T)})$ for all $T\in \mathcal{T}$ (the typical set of intermediate sets) and $(\pi, R)$ such that $R\in {\mathcal E_{\mathcal T}(\pi)}$, and extend $\widetilde{\query}(M, \mathbf{1}_{\pi(T)})$ to return an arbitrary element of $\pi(S\setminus T)$ for remaining tuples $(T, R, \pi(T),  \pi(S\setminus T))$. Note that $\widetilde{\query}$ defined in this way is a deterministic function once $T$, $R$, $\pi(T)$ and $\pi(S\setminus T)$ are fixed. 

Note that with probability at least $1-1/20$ over the choice of $\pi$ and $R$ one has $\query(M, \mathbf{1}_{\pi(T)})=\widetilde{\query}(M, \mathbf{1}_{\pi(T)})$ for all $T\in \mathcal{T}$, as required.

{\bf Step 3: arguing that $\mathcal{P}'$ succeeds with high probability.}  Choose $(\pi, R)$ uniformly at random. By {\bf Step 2} we have that with probability at least $1-1/20$ over this choice 
$\query(M, \mathbf{1}_{\pi(T)})=\widetilde{\query}(M, \mathbf{1}_{\pi(T)})$ for all $T\in \mathcal{T}$. At the same time we have by {\bf Step 1} that with probability at least $1-1/10$ over the choice of $\pi$ the idealized process $\widetilde{\mathcal{P}}$ 
succeeds in recovering $T_{i^*}$ and never leaves $\mathcal{T}$. Putting the two bounds together, we get that $\mathcal{P}'$ succeeds with probability at least $1-1/20-1/10>2/3$, showing the following theorem.

\begin{theorem}
For any $0<\delta<1/2$ and integer $n\ge 1$ with $\log \frac 1{\delta} < n/(4e)$, $\randcom^{\rightarrow,pub}_\delta(\ur^\subset) \ge \randcom^{\rightarrow,pub}_{1/3}(\aug_N)$ for $N = \Theta(\log\frac 1{\delta} \log^2 \frac n{\log \frac 1{\delta}})$.
\end{theorem}

\begin{corollary}
For any $0<\delta<1/2$ and integer $n\ge 1$, $\randcom^{\rightarrow,pub}_\delta(\ur^\subset) = \Omega(\min\{n, \log\frac 1{\delta} \log^2 \frac n{\log \frac 1{\delta}}\})$.
\end{corollary}

\subsection{Communication Lower Bound for $\ur_k^\subset$}\label{sec:aug-k}

The idea for lower bounding $\randcom^{\rightarrow,pub}_{\frac 15}(\ur_k^\subset)$ is as in Section~\ref{sec:aug-delta}, but slightly simpler. That is because for the protocol $\mathcal P'$ for $\aug_N$, Diane will not make adaptive queries to Bob in the protocol $\mathcal P$ for $\ur_k^\subset$. Rather, she will only make one query using Bob and will be able to decide $\aug_N$ with good probability from that single query. We  make use of the following lemma from~\cite{JowhariST11}, whose proof is similar to our analysis in Section~\ref{sec:aug-delta}.

\begin{lemma}{\cite{JowhariST11}}\label{lem:rand-ur}
Any public coin protocol for $\ur^\subset$ can be turned into one that outputs every index $i\in[n]$ with $x_i\neq y_i$ with the same probability. The number of bits sent, failure probability, and number of rounds do not change. Similarly, any $\ur_k^\subset$ protocol can be turned into one in which all subsets of $[n]$ of size $\min\{k, \|x-y\|_0\}$ on which $x, y$ differ are equally likely to be output.
\end{lemma}

Henceforth we assume $\mathcal P$ outputs random differing indices, which is without loss of generality by Lemma~\ref{lem:rand-ur}.

Again Charlie receives $z\in\{0,1\}^N$ and Diane receives $j^*$ and $z_{j^*+1},\ldots,z_N$ and they want to solve $\aug_N$. Charlie views his input as consisting of $L$ blocks for $L = c\log(n/k)$ for a sufficiently small constant $c\in(0,1)$, and the $i$th block for $i\in[L]$ specifies a set $S_i \in \mathcal S_{u_i,m}$ for $m = ck$ and $u_i = n/(100^i L)$. As before, for $c$ sufficiently small we have $N = \Theta(L\cdot k\log(n/k)) = \Theta(k\log^2(n/k))$. The bijection $A$ and set $S$ are defined exactly as in Section~\ref{sec:aug-delta}, and Charlie simulates Alice to send the message $M$ to Diane that Alice would have sent to Bob on input $\mathbf{1}_S$. Again, Diane knows $S_i$ for $i>i^*$, where $j^*$ lies in the $i^*$th block of bits. Diane's algorithm to produce her output is then described in Algorithm~\ref{algo:diane2}.

\begin{algorithm}%[H] 
  \caption{Behavior of Diane in $\mathcal P'$ for $\ur_k^\subset$.} \label{algo:diane2}
  \begin{algorithmic}[1]
    \Procedure{$\diane$}{$M$}
    \State $T \leftarrow \bigcup_{i=i^*+1}^L (\{i\} \times S_i \times [100^i])$
    \State $T_{i^*}\leftarrow \emptyset$
    \State $B\leftarrow \query(M, \mathbf{1}_T)$
    \For{$(i,a,r)\in B$}
      \If {$i=i^*$ and $a\notin T$}
        \State $T_{i^*} \leftarrow T_{i^*}\cup \{a\}$
      \EndIf
    \EndFor
    \If {$|T_{i^*}| < \frac m2$}
      \State \Return \textsf{Fail}
    \Else
      \State \Return the unique $S\in \mathcal S_{u_{i^*},m}$ with $T_{i^*}\subset S$ 
    \EndIf
    \EndProcedure
  \end{algorithmic}
\end{algorithm}

Recall Bob, when he succeeds, returns $\min\{k, |S\backslash T|\} = k$ uniformly random elements from $S\backslash T$. Meanwhile, $S_{i^*}$ only has $m = ck$ elements for some small constant $c$. As in Section~\ref{sec:aug-delta}, almost all of the support of $S\backslash T$ comes from items in block $i^*$, and hence we expect almost all our $k$ samples to come from (and be uniform in) items corresponding to elements of $S_{i^*}$. 

We now provide a formal analysis. Henceforth we condition on Bob succeeding, which happens with probability $4/5$. The number of elements in $S\backslash T$ corresponding to an element of $S_{i^*}$ is $100^{i^*} m$, whereas the number of elements corresponding to an element of $S_i$ for $i < i^*$ is
$$
m\cdot \sum_{i=1}^{i^*-1} 100^i = \frac m{99}\cdot (100^{i^*} - 1) < \frac m{99}\cdot 100^{i^*}
$$
Thus, we expect at most $k/99$ elements in $B$ to correspond to elements in $S_i$ for $i\neq i^*$, and the probability that we have at least $k/9$ such elements in $B$ is less than $1/10$ by Markov's inequality. We henceforth condition on having less than $k/9$ such elements in $B$. Now we know $B$ contains at least $8k/9$ elements corresponding to $S_{i^*}$, chosen uniformly from $S_{i^*}\times [100^i]$. For any given element $a\in S_{i^*}$, the probability that none of the elements in $B$ from $S_{i^*}$ correspond to $a$ is $(1 - 1/m)^{\frac 89 k} \le e^{-(8/9)k/m} < 1/30$ for $c$ sufficiently small (where $m = ck$). Thus the expected number of $a\in S_{i^*}$ not covered by $B$ is less than $m/30$. Thus the probability that fewer than $m/2$ elements are covered by $B$ is a most $1/15$ by Markov's inequality (and otherwise, Diane succeeds). Thus, the probability that Diane succeeds is at least $4/5\cdot 9/10 \cdot 14/15 > 2/3$. We have thus shown the following theorem.

\begin{theorem}
For any integers $1\le k \le n$, $\randcom^{\rightarrow,pub}_{\frac 15}(\ur_k^\subset) \ge \randcom^{\rightarrow,pub}_{\frac 13}(\aug_N)$ for $N = \Theta(k\log^2(n/k))$.
\end{theorem}

\begin{corollary}
For any integers $1\le k \le n$, $\randcom^{\rightarrow,pub}_{\frac 15}(\ur_k^\subset) = \Omega(\min\{n, k\log^2(n/k)\})$.
\end{corollary}

\begin{remark}
\textup{
One may wish to understand $\randcom^{\rightarrow,pub}_\delta(\ur_k^\subset)$ for $\delta$ near $1$ (or at least, larger than $1/2$). Such a lower bound is given in Theorem~\ref{thm:urk}. The proof given above as written would yield no lower bound in this regime for $\delta$ since $\aug$ is in fact easy when the failure probability is allowed to be least $1/2$ (Charlie can send no message at all, and Diane can simply guess $z_{j^*}$ via a coin flip). One can however get a handle on $\randcom^{\rightarrow,pub}_\delta(\ur_k^\subset)$ by instead directly reducing from the following variant of augmented indexing: Charlie receives $D\in \mathcal S_{u_1,m}\times \cdots \times \mathcal S_{u_L, m}$ and Diane receives $j^*\in[L]$ and $D_{j^*+1},\ldots,D_L$ and must output $D_{j^*}$, where the $u_i$ are as above. One can show that unless Charlie sends almost his entire input, Diane cannot have success probability significantly better than random guessing (which has success probability $O(\max_{i\in L} 1/|\mathcal S_{u_i, m}|)$). The proof is nearly identical to the analysis of augmented indexing over large domains \cite{ErgunJS10,JayramW13}. Indeed, the problem is even almost identical, except that here we consider Charlie receiving a vector whose entries come from different alphabet sizes (since the $|\mathcal S_{u_i,m}|$ are different), whereas in \cite{ErgunJS10,JayramW13} all the entries come from the same alphabet.
}
\end{remark}

\section*{Acknowledgments}
Initially the authors were focused on proving optimal lower bounds for samplers, but we thank Vasileios Nakos for pointing out that our $\ur^\subset$ lower bound immediately implies a tight lower bound for finding a duplicate in data streams as well. Also, initially our proof of Lemma~\ref{lem:information} incurred an additive $1$ in the numerator of the right hand side of \eqref{eqn:adaptivity}. This is clearly suboptimal for small $I(X; Y)$ (for example, consider $I(X; Y) = 0$, in which case the right hand side should be $\delta$ and not $1/\log(1/\delta)$)). We thank T.S.\ Jayram for pointing out that a slight modification of our proof could actually replace the additive $1$ with the binary entropy function (and also for showing us a different proof of this lemma, which resembles the standard proof of Fano's inequality).

\bibliographystyle{alpha}

\begin{thebibliography}{GMWW14}

\bibitem[AGM12a]{AhnGM12a}
Kook~Jin Ahn, Sudipto Guha, and Andrew McGregor.
\newblock Analyzing graph structure via linear measurements.
\newblock In {\em Proceedings of the 23$^{rd}$ {ACM-SIAM} Symposium on Discrete
  Algorithms (SODA)}, pages 459--467, 2012.

\bibitem[AGM12b]{AhnGM12b}
Kook~Jin Ahn, Sudipto Guha, and Andrew McGregor.
\newblock Graph sketches: sparsification, spanners, and subgraphs.
\newblock In {\em Proceedings of the 31$^{st}$ {ACM} {SIGMOD-SIGACT-SIGART}
  Symposium on Principles of Database Systems (PODS)}, pages 5--14, 2012.

\bibitem[AGM13]{AhnGM13}
Kook~Jin Ahn, Sudipto Guha, and Andrew McGregor.
\newblock Spectral sparsification in dynamic graph streams.
\newblock In {\em Proceedings of the 16$^{th}$ International Workshop on
  Approximation Algorithms for Combinatorial Optimization Problems (APPROX)},
  pages 1--10, 2013.

\bibitem[AKL17]{AssadiKL17}
Sepehr Assadi, Sanjeev Khanna, and Yang Li.
\newblock On estimating maximum matching size in graph streams.
\newblock In {\em Proceedings of the 28$^{th}$ Annual {ACM-SIAM} Symposium on
  Discrete Algorithms (SODA)}, pages 1723--1742, 2017.

\bibitem[AKLY16]{AssadiKLY16}
Sepehr Assadi, Sanjeev Khanna, Yang Li, and Grigory Yaroslavtsev.
\newblock Maximum matchings in dynamic graph streams and the simultaneous
  communication model.
\newblock In {\em Proceedings of the 27$^{th}$ Annual {ACM-SIAM} Symposium on
  Discrete Algorithms (SODA)}, pages 1345--1364, 2016.

\bibitem[AKO11]{AndoniKO11}
Alexandr Andoni, Robert Krauthgamer, and Krzysztof Onak.
\newblock Streaming algorithms via precision sampling.
\newblock In {\em Proceedings of the 52$^{nd}$ Annual IEEE Symposium on
  Foundations of Computer Science (FOCS)}, pages 363--372, 2011.

\bibitem[BHNT15]{BhattacharyaHNT15}
Sayan Bhattacharya, Monika Henzinger, Danupon Nanongkai, and Charalampos~E.
  Tsourakakis.
\newblock Space- and time-efficient algorithm for maintaining dense subgraphs
  on one-pass dynamic streams.
\newblock In {\em Proceedings of the 47$^{th}$ Annual {ACM} on Symposium on
  Theory of Computing (STOC)}, pages 173--182, 2015.

\bibitem[BS15]{BuryS15}
Marc Bury and Chris Schwiegelshohn.
\newblock Sublinear estimation of weighted matchings in dynamic data streams.
\newblock In {\em Proceedings of the 23$^{rd}$ Annual European Symposium on
  Algorithms (ESA)}, pages 263--274, 2015.

\bibitem[CCE{\etalchar{+}}16]{ChitnisCEHMMV16}
Rajesh Chitnis, Graham Cormode, Hossein Esfandiari, MohammadTaghi Hajiaghayi,
  Andrew McGregor, Morteza Monemizadeh, and Sofya Vorotnikova.
\newblock Kernelization via sampling with applications to finding matchings and
  related problems in dynamic graph streams.
\newblock In {\em Proceedings of the 27$^{th}$ Annual {ACM-SIAM} Symposium on
  Discrete Algorithms (SODA)}, pages 1326--1344, 2016.

\bibitem[CCHM15]{ChitnisCHM15}
Rajesh~Hemant Chitnis, Graham Cormode, Mohammad~Taghi Hajiaghayi, and Morteza
  Monemizadeh.
\newblock Parameterized streaming: Maximal matching and vertex cover.
\newblock In {\em Proceedings of the 26$^{th}$ Annual {ACM-SIAM} Symposium on
  Discrete Algorithms (SODA)}, pages 1234--1251, 2015.

\bibitem[CF14]{CormodeF14}
Graham Cormode and Donatella Firmani.
\newblock A unifying framework for $\ell_0$-sampling algorithms.
\newblock {\em Distributed and Parallel Databases}, 32(3):315--335, 2014.
\newblock Preliminary version in ALENEX 2013.

\bibitem[CK04]{CoppersmithK04}
Don Coppersmith and Ravi Kumar.
\newblock An improved data stream algorithm for frequency moments.
\newblock In {\em Proceedings of the 15$^{th}$ Annual {ACM-SIAM} Symposium on
  Discrete Algorithms (SODA)}, pages 151--156, 2004.

\bibitem[CMR05]{CormodeMR05}
Graham Cormode, S.~Muthukrishnan, and Irina Rozenbaum.
\newblock Summarizing and mining inverse distributions on data streams via
  dynamic inverse sampling.
\newblock In {\em Proceedings of the 31$^{st}$ International Conference on Very
  Large Data Bases (VLDB)}, pages 25--36, 2005.

\bibitem[DM16]{DinurM16}
Irit Dinur and Or~Meir.
\newblock Toward the {KRW} composition conjecture: Cubic formula lower bounds
  via communication complexity.
\newblock In {\em Proceedings of the 31$^{st}$ Conference on Computational
  Complexity (CCC)}, pages 3:1--3:51, 2016.

\bibitem[EHW16]{EsfandiariHW16}
Hossein Esfandiari, MohammadTaghi Hajiaghayi, and David~P. Woodruff.
\newblock Brief announcement: Applications of uniform sampling: Densest
  subgraph and beyond.
\newblock In {\em Proceedings of the 28$^{th}$ {ACM} Symposium on Parallelism
  in Algorithms and Architectures (SPAA)}, pages 397--399, 2016.

\bibitem[EIRS91]{EIRS91}
Jack Edmonds, Russell Impagliazzo, Steven Rudich, and Jiri Sgall.
\newblock Communication complexity towards lower bounds on circuit depth.
\newblock In {\em Proceedings of the 32$^{nd}$ Annual IEEE Symposium on the
  Foundations of Computer Science (FOCS)}, pages 249--257, 1991.

\bibitem[EJS10]{ErgunJS10}
Funda Erg{\"{u}}n, Hossein Jowhari, and Mert Sa{\u{g}}lam.
\newblock Periodicity in streams.
\newblock In {\em Proceedings of the 14$^{th}$ International Workshop on
  Randomization and Approximation Techniques in Computer Science (RANDOM)},
  pages 545--559, 2010.

\bibitem[FIS08]{FrahlingIS08}
Gereon Frahling, Piotr Indyk, and Christian Sohler.
\newblock Sampling in dynamic data streams and applications.
\newblock {\em Int. J. Comput. Geometry Appl.}, 18(1/2):3--28, 2008.
\newblock Preliminary version in SOCG 2005.

\bibitem[FT16]{FarachColtonT16}
Martin Farach{-}Colton and Meng{-}Tsung Tsai.
\newblock Tight approximations of degeneracy in large graphs.
\newblock In {\em Proceedings of the 12$^{th}$ Latin American Symposium on
  Theoretical Informatics (LATIN)}, pages 429--440, 2016.

\bibitem[GKKT15]{GibbKKT15}
David Gibb, Bruce~M. Kapron, Valerie King, and Nolan Thorn.
\newblock Dynamic graph connectivity with improved worst case update time and
  sublinear space.
\newblock {\em CoRR}, abs/1509.06464, 2015.

\bibitem[GMT15]{GuhaMT15}
Sudipto Guha, Andrew McGregor, and David Tench.
\newblock Vertex and hyperedge connectivity in dynamic graph streams.
\newblock In {\em Proceedings of the 34$^{th}$ {ACM} Symposium on Principles of
  Database Systems (PODS)}, pages 241--247, 2015.

\bibitem[GMWW14]{GavinskyMWW14}
Dmitry Gavinsky, Or~Meir, Omri Weinstein, and Avi Wigderson.
\newblock Toward better formula lower bounds: an information complexity
  approach to the {KRW} composition conjecture.
\newblock In {\em Proceedings of the 46$^{th}$ Annual ACM Symposium on Theory
  of Computing (STOC)}, pages 213--222, 2014.

\bibitem[GR09]{GopalanR09}
Parikshit Gopalan and Jaikumar Radhakrishnan.
\newblock Finding duplicates in a data stream.
\newblock In {\em Proceedings of the 20$^{th}$ Annual {ACM-SIAM} Symposium on
  Discrete Algorithms (SODA)}, pages 402--411, 2009.

\bibitem[HPP{\etalchar{+}}15]{HegemanPPSS15}
James~W. Hegeman, Gopal Pandurangan, Sriram~V. Pemmaraju, Vivek~B. Sardeshmukh,
  and Michele Scquizzato.
\newblock Toward optimal bounds in the congested clique: Graph connectivity and
  {MST}.
\newblock In {\em Proceedings of the 34$^{th}$ Annual {ACM} Symposium on
  Principles of Distributed Computing (PODC)}, pages 91--100, 2015.

\bibitem[HW90]{HastadW90}
Johan H{\aa}stad and Avi Wigderson.
\newblock Composition of the universal relation.
\newblock In {\em Proceedings of a {DIMACS} Workshop on Advances In
  Computational Complexity Theory}, pages 119--134, 1990.

\bibitem[JST11]{JowhariST11}
Hossein Jowhari, Mert Sa{\u{g}}lam, and G{\'a}bor Tardos.
\newblock Tight bounds for {Lp} samplers, finding duplicates in streams, and
  related problems.
\newblock In {\em Proceedings of the 30$^{th}$ ACM SIGMOD-SIGACT-SIGART
  Symposium on Principles of Database Systems (PODS)}, pages 49--58. ACM, 2011.

\bibitem[JW13]{JayramW13}
T.~S. Jayram and David~P. Woodruff.
\newblock Optimal bounds for {Johnson}-{Lindenstrauss} transforms and streaming
  problems with subconstant error.
\newblock {\em {ACM} Trans. Algorithms}, 9(3):26:1--26:17, 2013.

\bibitem[KKM13]{KapronKM13}
Bruce~M. Kapron, Valerie King, and Ben Mountjoy.
\newblock Dynamic graph connectivity in polylogarithmic worst case time.
\newblock In {\em Proceedings of the 24$^{th}$ Annual {ACM-SIAM} Symposium on
  Discrete Algorithms (SODA)}, pages 1131--1142, 2013.

\bibitem[KLM{\etalchar{+}}14]{KapralovLMMS14}
Michael Kapralov, Yin~Tat Lee, Cameron Musco, Christopher Musco, and Aaron
  Sidford.
\newblock Single pass spectral sparsification in dynamic streams.
\newblock In {\em Proceedings of the 55$^{th}$ {IEEE} Annual Symposium on
  Foundations of Computer Science (FOCS)}, pages 561--570, 2014.

\bibitem[KNW10]{KaneNW10}
Daniel~M. Kane, Jelani Nelson, and David~P. Woodruff.
\newblock An optimal algorithm for the distinct elements problem.
\newblock In {\em Proceedings of the 29$^{th}$ ACM SIGMOD-SIGACT-SIGART
  Symposium on Principles of Database Systems (PODS)}, pages 41--52, 2010.

\bibitem[Kon15]{Konrad15}
Christian Konrad.
\newblock Maximum matching in turnstile streams.
\newblock In {\em Proceedings of the 23$^{rd}$ Annual European Symposium on
  Algorithms (ESA)}, pages 840--852, 2015.

\bibitem[KRW95]{KarchmerRW95}
Mauricio Karchmer, Ran Raz, and Avi Wigderson.
\newblock Super-logarithmic depth lower bounds via the direct sum in
  communication complexity.
\newblock {\em Computational Complexity}, 5(3-4):191--204, 1995.

\bibitem[KW90]{KarchmerW90}
Mauricio Karchmer and Avi Wigderson.
\newblock Monotone circuits for connectivity require super-logarithmic depth.
\newblock {\em {SIAM} J. Discrete Math.}, 3(2):255--265, 1990.

\bibitem[McG14]{McGregor14}
Andrew McGregor.
\newblock Graph stream algorithms: a survey.
\newblock {\em {SIGMOD} Record}, 43(1):9--20, 2014.

\bibitem[MNSW98]{MiltersenNSW98}
Peter~Bro Miltersen, Noam Nisan, Shmuel Safra, and Avi Wigderson.
\newblock On data structures and asymmetric communication complexity.
\newblock {\em J. Comput. Syst. Sci.}, 57(1):37--49, 1998.

\bibitem[MTVV15]{McGregorTVV15}
Andrew McGregor, David Tench, Sofya Vorotnikova, and Hoa~T. Vu.
\newblock Densest subgraph in dynamic graph streams.
\newblock In {\em Proceedings of the 40$^{th}$ International Symposium on
  Mathematical Foundations of Computer Science (MFCS)}, pages 472--482, 2015.

\bibitem[Mut05]{Muthukrishnan05}
S.~Muthukrishnan.
\newblock {Data Streams: Algorithms and Applications}.
\newblock {\em Foundations and Trends in Theoretical Computer Science},
  1(2):117--236, 2005.

\bibitem[MW10]{MonemizadehW10}
Morteza Monemizadeh and David~P. Woodruff.
\newblock 1-pass relative-error $l_p$-sampling with applications.
\newblock In {\em Proceedings of the 21$^{st}$ Annual {ACM-SIAM} Symposium on
  Discrete Algorithms (SODA)}, pages 1143--1160, 2010.

\bibitem[NPW17]{NelsonPW17}
Jelani Nelson, Jakub Pachocki, and Zhengyu Wang.
\newblock Optimal lower bounds for universal relation, samplers, and finding
  duplicates.
\newblock {\em CoRR}, abs/1703.08139, March 2017.

\bibitem[PRS16]{Pandurangan0S16}
Gopal Pandurangan, Peter Robinson, and Michele Scquizzato.
\newblock Fast distributed algorithms for connectivity and {MST} in large
  graphs.
\newblock In {\em Proceedings of the 28$^{th}$ {ACM} Symposium on Parallelism
  in Algorithms and Architectures (SPAA)}, pages 429--438, 2016.

\bibitem[Tar07]{Tarui07}
Jun Tarui.
\newblock Finding a duplicate and a missing item in a stream.
\newblock In {\em Proceedings of the 4$^{th}$ Annual Conference on Theory and
  Applications of Models of Computation (TAMC)}, pages 128--135, 2007.

\bibitem[TZ97]{TardosZ97}
G{\'{a}}bor Tardos and Uri Zwick.
\newblock The communication complexity of the universal relation.
\newblock In {\em Proceedings of the 12$^{th}$ Annual {IEEE} Conference on
  Computational Complexity (CCC)}, pages 247--259, 1997.

\bibitem[Wan15]{Wang15}
Zhengyu Wang.
\newblock An improved randomized data structure for dynamic graph connectivity.
\newblock {\em CoRR}, abs/1510.04590, 2015.

\end{thebibliography}

\newcommand{\etalchar}[1]{$^{#1}$}

\appendix

\section{Appendix}

\subsection{A tight upper bound for $\randcom^{\rightarrow,pub}_\delta(\ur_k)$}

In \cite[Proposition 1]{JowhariST11} it is shown that $\randcom^{\rightarrow,pub}_\delta(\ur_k) = O(\min\{n,t\log^2 n\})$ for $t = \max\{k,\log(1/\delta)\}$. Here we show that a minor modification of their protocol in fact shows the correct complexity $\randcom^{\rightarrow,pub}_\delta(\ur_k) = O(\min\{n,t\log^2(n/t)\})$, which given our new lower bound, is optimal up to a constant factor for the full range of $n,k,\delta$ as long as $\delta$ is bounded away from $1$.

Recall Alice and Bob receive $x, y\in\{0,1\}^n$, respectively, and share a public random string. Alice must send a single message $M$ to Bob, from which Bob must recover $\min\{k, \|x-y\|_0\}$ indices $i\in[n]$ for which $x_i\neq y_i$. Bob is allowed to fail with probability $\delta$. The fact that $\randcom^{\rightarrow,pub}_\delta(\ur_k) \le n$ is obvious: Alice can simply send the message $M = x$, and Bob can then succeed with failure probability $0$. We thus now show $\randcom^{\rightarrow,pub}_{e^{-ck}}(\ur_k) \le k\log^2(n/k)$ for some constant $c>0$, which completes the proof of the upper bound. We assume $k\le n/2$ (otherwise, Alice sends $x$ explicitly).

As mentioned, the protocol we describe is nearly identical to one in \cite{JowhariST11} (see also \cite{CormodeF14}). We will describe the new protocol here, then point out the two minor modifications that improve the $O(k\log^2 n)$ bound to $O(k\log^2(n/k))$ in Remark~\ref{rem:recov}. We first need the following lemma.

\begin{lemma}\label{lem:sparse-recov}
Let $\F_q$ be a finite field and $n>1$ an integer. Then for any $1\le k\le \frac n2$, there exists $\Pi_k\in \F_q^{m\times n}$ for $m = O(k\log_q(qn/k))$ s.t.\ for any $w\neq w'\in\F_q^n$ with $\|w\|_0, \|w'\|_0 \le k$, $\Pi_k w \neq \Pi_k w'$.
\end{lemma}
\begin{proof}
The proof is via the probabilistic method. $\Pi_k w = \Pi_k w'$ iff $\Pi_k (w - w') = 0$. Note $v = w-w'$ has $\|v\|_0 \le 2k$. Thus it suffices to show that such a $\Pi_k$ exists with no $(2k)$-sparse vector in its kernel. The number of vectors $v\in\F_q^n$ with $\|v_0\| \le 2k$ is at most $\binom{n}{2k}\cdot q^{2k}$. For any fixed $v$, $\Pr(\Pi_k v = 0) = q^{-m}$. Thus 
$$\Pr(\exists v, \|v\|_0 \le 2k: \Pi_k v = 0) \le \binom{n}{2k}\cdot q^{2k} \cdot q^{-m}$$ 
by a union bound. The above is strictly less than $1$ for $m > 2k + \log_q\binom{n}{2k}$, yielding the claim.
\end{proof}

\begin{corollary}\label{cor:ksparse}
Let $\F_q$ be a finite field and $n>1$ an integer. Then for any $1\le k\le \frac n2$, there exists $\Pi_k\in \F_q^{m\times n}$ for $m = O(k\log_q(qn/k))$ together with an algorithm $\mathcal{R}$ such that for any $w\in\F_q^n$ with $\|w\|_0 \le k$, $\mathcal{R}(\Pi_k w) = w$.
\end{corollary}
\begin{proof}
Given Lemma~\ref{lem:sparse-recov}, a simple such $\mathcal{R}$ is as follows. Given some $y = \Pi_k w^*$ with $\|w^*\|_0 \le k$, $\mathcal{R}$ loops over all $w$ in $\F_q^n$ with $\|w\|_0 \le k$ and outputs the first one it finds for which $\Pi_k w = y$.
\end{proof}

The protocol for $\ur_k$ is now as follows. Alice and Bob use public randomness to pick commonly known random functions $h_0,\ldots,h_L:[n]\rightarrow\{0,1\}$ for $L = \lfloor\log_2(n/k)\rfloor$, such that for any $i\in[n]$ and for any $j$, $\Pr(h_j(i) = 1) = 2^{-j}$. They also agree on a matrix $\Pi_{16k}$ and $\mathcal{R}$ as described in Corollary~\ref{cor:ksparse} for a sufficiently large constant $C>0$ to be determined later, with $q = 3$. Thus $\Pi_{16k}$ has $m = O(k\log(n/k))$ rows. Alice then computes $v_j = \Pi_{16k} x|_{h_j^{-1}(1)}$ for $j=0,\ldots,L$ where $v_j\in\F_q^m$, and her message to Bob is $M = (v_0,\ldots,v_L)$. For $S\subseteq [n]$ and $x$ an $n$-dimensional vector, $x|_S$ denotes the $n$-dimensional vector with $(x|_S)_i = x_i$ for $i\in S$, and $(x|_S)_i = 0$ for $i\notin S$. Note Alice's message $M$ is $O(k\log^2(n/k))$ bits, as desired. Bob then executes the following algorithm and outputs the returned values.

\begin{algorithm}[H] 
  \caption{Bob's algorithm in the $\ur_k$ protocol.} \label{algo:bob-protocol}
  \begin{algorithmic}[1]
    \Procedure{Bob}{$v_0,\ldots,v_L$}
    \For {$j=L,L-1,\ldots,0$}
      \State $v_j \leftarrow v_j - \Pi_{16k} y|_{h_j^{-1}(1)}$
      \State $w_j\leftarrow \mathcal{R}(v_j)$
      \If {$\|w_j\|_0 \ge k$ or $j=0$}
      \State \Return an arbitrary $\min\{k, \|w_j\|_0\}$ elements from $\supp(w_j)$
      \EndIf
    \EndFor
    \EndProcedure
  \end{algorithmic}
\end{algorithm}

The correctness analysis is then as follows, which is nearly the same as the $\ell_0$-sampler of \cite{JowhariST11}. If Alice's input is $x$ and Bob's is $y$, let $a = x-y \in \{-1,0,1\}^n$, so that $a$ can be viewed as an element of $\F_3^n$. Also let $a_j = a|_{h_j^{-1}(1)}$. Then $\E \|v_j\|_0 = \|a\|_0\cdot 2^{-j}$, and since $0\le \|a\|_0 \le n$, there either (1) exists a unique $0\le j^*\le L$ such that $2k\le \E\|a_j\|_0\cdot 2^{-j^*}< 4k$, or (2) $\|a\|_0 < 2k$ (in which case we define $j^* = 0$). Let $\mathcal{E}$ be the event that $\|a_j\|_0 \le 16k$ simultaneously for all $j\le j^*$. Let $\mathcal{F}$ be the event that {\it either} we are in case (2), or we are in case (1) and $\|a_{j^*}\|_0 \ge k$ holds. Note that conditioned on $\mathcal{E}, \mathcal{F}$ both occurring, Bob succeeds by Corollary~\ref{cor:ksparse}.

We now just need to show $\Pr(\neg\mathcal{E} \wedge \neg\mathcal{F}) < e^{-\Omega(k)}$. We use the union bound. First, consider $\mathcal{F}$. If $j^* = 0$, then $\Pr(\neg\mathcal{F}) = 0$. If $j^*\neq 0$, then $\Pr(\neg\mathcal{F}) \le \Pr(\|a_{j^*}\|_0 < \frac 12 \cdot\E\|a_{j^*}\|_0)$, which is $e^{-\Omega(k)}$ by the Chernoff bound since $\E\|a_{j^*}\|_0 = \Theta(k)$. Next we bound $\Pr(\neg \mathcal{E})$. For $j\ge j^*$, we know $\E\|a_j\|_0 \le 4k/2^{j-j^*}$. Thus, letting $\mu$ denote $\E\|a_j\|_0$, 
\begin{equation}
\Pr(\|a_j\|_0 > 16k) < \left(\frac{e^{\frac{16k}{\mu} - 1}}{(\frac{16k}{\mu})^{\frac{16k}{\mu}}}\right)^\mu < \left(\frac{16k}{\mu}\right)^{-\Omega(k)} < (e^{-Ck})^{j-j^*}\label{eqn:geometric}
\end{equation}
for some constant $C>0$ by the Chernoff bound and the fact that $16k/\mu \ge 4 > e$. Recall that the Chernoff bound states that for $X$ a sum of independent Bernoullis,
$$
\forall \delta > 0,\ \Pr(X > (1+\delta) \E X) < \left(\frac{e^\delta}{(1+\delta)^{1+\delta}}\right)^{\E X} .
$$
Then by a union bound over $j\ge j^*$ and applying \eqref{eqn:geometric},
$$
\Pr(\neg \mathcal{E}) = \Pr(\exists j\ge j^*: \|a_j\|_0 > 16k) < \sum_{j=j^*}^\infty (e^{-Ck})^{j-j^*} = O(e^{-Ck}) .
$$

\begin{remark}\label{rem:recov}
\textup{
As already mentioned, the protocol given above and the one described in \cite{JowhariST11} using $O(k\log^2 n)$ bits differ in minor points. First: the protocol there used $\lfloor\log_2 n\rfloor$ different hash functions $h_j$, but as seen above, only $\lfloor \log_2(n/k)\rfloor$ are needed. This already improves one $\log n$ factor to $\log(n/k)$. The other improvement comes from replacing the $k$-sparse recovery structure with $2k$ rows used in \cite{JowhariST11} with our Corollary~\ref{cor:ksparse}. Note the matrix $\Pi_k$ in our corollary has even {\it more} rows, but the key point is that the bit complexity is improved. Whereas using a $k$-sparse recovery scheme as described in \cite{JowhariST11} would use $2k$ linear measurements of a $k$-sparse vector $w\in\{-1,0,1\}^n$ with $\log n$ bits per measurement (for a total of $O(k\log n)$ bits), we use $O(k\log(n/k))$ measurements with only $O(1)$ bits per measurement. The key insight is that we can work over $\F_3^n$ instead of $\R^n$ when the entries of $w$ are in $\{-1,0,1\}$, which leads to our slight improvement.
}
\end{remark}

\subsection{Proof of the existence of the desired $\mathcal S_{u,m}$}\label{sec:code}
\noindent \textbf{Lemma~\ref{lem:code} (restated).}
For any integers $u\ge 1$ and $1\le m\le u/(4e)$, there exists a collection $\mathcal S_{u,m} \subset \binom{[u]}m$ with $\log |\mathcal{S}_{u,m}| = \Theta(m\log(u/m))$ such that for all $S\neq S'\in \mathcal S_{u,m}$, $|S\cap S'| < m/2$.
\begin{proof}
The proof is via the probabilistic method. We pick $S_1,\ldots,S_N$ independently, each one uniformly at random from $\binom{[u]}m$. Fix $i\neq j\in[N]$. Imagine $S_i$ being fixed and picking the $m$ elements of $S_j$ one by one. Let $X_k$ denote the indicator random variable for the event that the $k$th element picked is also in $S_i$. Then $|S_i\cap S_j| = \sum_{k=1}^m X_k$, and we set $\mu:= \E |S_i\cap S_j|$, which is $m^2/u$ by linearity of expectation. We have $\Pr(|S_i \cap S_j| \ge m/2) = \Pr(|S_i\cap S_j| \ge (1+\delta)\mu)$ for $\delta = u/(2m) - 1$. The $X_k$ are not independent, but they are negatively dependent. Thus the Chernoff bound yields
$$
\Pr(|S_i\cap S_j| \ge (1+\delta)\mu) \le \left(\frac{e^{\delta}}{(1+\delta)^{1+\delta}}\right)^\mu \le \left(\frac{e^{\frac u{2m} - 1}}{(\frac u{2m})^{\frac u{2m}}}\right)^{m^2/u} \le \left(\frac u{2em}\right)^{-\frac m2} .
$$
Setting $N = \sqrt{(u/(2em))^{m/2} - 1}$ so that ${N \choose 2}\leq N^2=(u/(2em))^{m/2} - 1$, by a union bound with positive probability $|S_i\cap S_j| < m/2$ for all $i\neq j$, simultaneously, as desired. Note for this choice of $N$, we have $\log|\mathcal S_{u,m}| = \log N = \Theta(m\log(u/m))$.
\end{proof}

\end{document}